%% file: main_arxiv.tex
\documentclass{article}[11pt]

\usepackage{xcolor}
\usepackage{comment}
\usepackage{amsmath,amsfonts,amsthm,prodint}
\usepackage{graphicx}
\usepackage{booktabs}
\usepackage{natbib}
\usepackage{appendix}
\usepackage{comment}
\usepackage{url}            
\usepackage{bibunits}
\usepackage{hyperref}       
\usepackage[capitalize,nameinlink]{cleveref}

\newtheorem{prop}{Proposition}
\newtheorem{lem}{Lemma}

\newtheorem{assp}{Assumption}
\newtheorem{thm}{Theorem}
\newtheorem{remark}{Remark}

\newcommand{\ind}{\perp\!\!\!\!\perp}
\newcommand{\M}{\mathcal{M}}
\newcommand{\full}{*}
\newcommand{\mO}{\mathcal{O}}
\renewcommand{\P}{\mathbb{P}} 
\newcommand{\bP}{\bar{\P}}
\newcommand{\Pf}{\P^{\full}}
\newcommand{\E}{\mathbb{E}} 
\newcommand{\bE}{\bar{\E}}
\newcommand{\Ef}{\E^{\full}}
\newcommand{\PP}{\mathbb{P}} 

\newcommand{\sing}{\mathrm{single}}

\newcommand\mC{\mathcal{C}}
\renewcommand\d{\mathrm{d}}

\usepackage[margin=1.0in]{geometry}
\title{Causal inference for the expected number of recurrent events in the presence of a terminal event}
\date{}
\author{Benjamin R. Baer \\
School of Mathematics and Statistics \\
University of St Andrews, St Andrews,  KY16 9SS, Scotland \\[.5ex]
Trang Bui \\
Department of Biostatistics and Computational Biology \\
University of Rochester, Rochester, New York, U.S.A. \\[.5ex]
Daniel Mork \\
Department of Biostatistics, Harvard T.H. Chan School of Public Health, Boston, MA, U.S.A. \\[.5ex]
Robert L. Strawderman \footnote{Corresponding author: \href{robert\_strawderman@urmc.rochester.edu}{robert\_strawderman@urmc.rochester.edu}} \\
Department of Biostatistics and Computational Biology \\
University of Rochester, Rochester, New York, U.S.A. \\[.5ex]
Ashkan Ertefaie \\
Department of Biostatistics, Epidemiology \& Informatics \\ University of Pennsylvania, Philadelphia, Pennsylvania, U.S.A.\\[.5ex]
}

\begin{document}
\maketitle



\begin{abstract}
While recurrent event analyses have been extensively studied, limited attention has been given to causal inference within the framework of recurrent event analysis. We develop a multiply robust estimation framework  for causal inference in recurrent event data  with a terminal failure event. We define our estimand as the vector comprising both the expected number of recurrent events and the failure survival function evaluated along a sequence of landmark times. We show that the estimand can be identified under a weaker condition than conditionally independent censoring and derive the associated class of influence functions under general censoring and failure distributions (i.e., without assuming absolute continuity). We propose a particular estimator within this class for further study, conduct comprehensive simulation studies to evaluate the small-sample performance of our estimator, and illustrate the proposed estimator using a large Medicare dataset to assess the causal effect of PM$_{2.5}$ on recurrent cardiovascular hospitalization. \\[1ex]

\end{abstract}


\input{main_body.tex}

\newpage

\input{supp_doc.tex}

\end{document}

%% file: main_body.tex
\begin{bibunit}[biom]

\section{Introduction}

Survival analysis is a fundamental tool for analyzing time-to-event data in diverse fields, including medicine, epidemiology, and the social sciences \citep{fleming1991counting}. It involves studying the timing of events such as death, disease progression, or system failure, and understanding the factors that influence their occurrence. In many applications, individuals may experience recurrent events, where multiple instances of the event occur over time \citep{cook2007statistical}. This introduces additional complexity, as earlier events may affect the likelihood and timing of subsequent ones. As researchers increasingly seek to draw causal conclusions about the effects of treatments on survival outcomes, there is a growing need for methods that support valid causal inference. This work addresses a problem situated at the intersection of these three research areas.

There is an existing literature on the analysis of recurrent event data in the presence of a failure (i.e., terminal event) time; see \citet[Section 6.6]{cook2007statistical} for one review. \citet{cook1997marginal} developed some estimands and both inverse-probability weighted estimators and proportional intensity regression estimators. \citet{zhao1997consistent, lin1997estimating} studied inverse-probability weighted estimators in a related setting. \citet{ghosh2000nonparametric, cook2009robust} further study inverse-probability weighted estimation of the expected number of recurrent events before failure, an estimand of central importance in this work that generalizes the cumulative incidence function from competing-event settings. \citet{strawderman2000estimating} study related inverse-probability weighted and also augmented estimators for the same estimand.  However, the focus on causal inference within the framework of recurrent event analysis has been limited. To the authors' knowledge,
\citet{schaubel2010estimating} initiated the study of the expected number of recurrent events before failure from a causal inference perspective and proposed ``doubly'' inverse-probability weighted and imputation estimators. \citet{janvin2022causal} further study inverse-probability weighted estimators with a focus on a mediation perspective on the relationship between failure and the recurrent process. \citet{jensen2016marginal,su2020doubly} study causal inference for recurrent events but do not allow for the possibility of termination due to failure. 

Our primary contribution is to develop a multiply robust estimation framework for a specific recurrent event estimand to be defined later. Importantly, we do not make any absolute continuity assumptions and work with essentially arbitrary probability distributions for the observed data, including the required latent distributions of failure and censoring. First, we establish a set of identifiability conditions for the causal recurrent event problem (see Section~\ref{sec:estimand:car}). Second, we identify the nuisance parameters and the desired estimand under a condition weaker than conditionally independent censoring (see Section~\ref{sec:estimand:identif:results}). Third, we characterize the corresponding class of influence functions (Section~\ref{sec:estimation:class}) and propose a multiply robust estimator for both the expected number of counterfactual recurrent events before failure and characterize the corresponding remainder terms, without assuming absolute continuity (Section~\ref{sec:estimation:eif}). Fourth, we review existing semiparametric estimators for this estimand (Section~\ref{sec:litreview:recur}) and show that a previously proposed doubly robust estimator does not
possess this desirable property (Section~\ref{sec:litreview:recur:existing-est}). The remainder of the paper evaluates estimator performance through
simulation studies (Section~\ref{sec:numeric}), illustrates the methods by studying the effect of air pollution exposure on cardiovascular disease
hospitalization (Section~\ref{sec:example}), and closes with a discussion (Section~\ref{sec:conc}).

\section{Notation}

Let $L$ denote a vector consisting of baseline variables for an individual. Assuming there is a binary treatment taking on the values zero and one, let $T^{\full}_0, T^{\full}_1$ be the potential failure times had an individual not been or been treated, respectively. Similarly, let the right-continuous processes $N^*_0(\cdot), N^*_1(\cdot)$ denote the corresponding potential recurrent event process had an individual not been or been treated, respectively. We assume that recurrent events do not occur past, or at, failure so that $N^*_a(t) = N^*_a(T^*_a) = N^*_a(T^*_a-)$ 
for any $t \geq T^*_a$ and each $a=0,1.$ 
Finally, define the full data for an individual as $(L, T^{\full}_0, T^{\full}_1, \{ N^*_0(t), N^*_1(t), 0 \leq t < \infty\}).$ For later use, we also define the hypothetical unstopped recurrent event process $N^{**}_a(\cdot)$ so that $N^*_a(t) = N^{**}_a \bigl\{ \min\{t, T^*_a\} \bigr\}$ while $N^{**}_a(\cdot)$ and $T^*_a$ are variationally independent. 

Define the coarsening variables $A, C^{\full}_0, C^{\full}_1$. The treatment (i.e., exposure) variable $A$ is $1$ when treatment is assigned to an individual and $0$ otherwise. The potential censoring variables $C^{\full}_0,C^{\full}_1$ are the times that observation ceases had the exposure been $a$ (i.e., for reasons other than failure). Define the joint distribution of the full and coarsening variables as $\Pf$. For $a=0,1$ define the conditional counterfactual survival function for failure as $H^{\full}(u; a, l) = \Pf (T^{\full}_a > u \mid L=l)$, and define $F^{\full}(u, t; a, l) = \Ef \{ I(T^{\full}_a > u) N^{\full}_a(t) \mid L=l \}$ so that $F^{\full}(0, t; a, l)$ is the conditional counterfactual expected number of recurrent events. Define the propensity score $\pi(a; l) = \PP(A=a \mid L=l)$ and the conditional counterfactual censoring survival function $K^{\full}(u; a, l) = \Pf(C^{\full}_a > u \mid L=l)$. 

The observed data $\mO$ on one subject is defined by limiting the availability of the full data through the  mapping $\Phi(L,T^{\full}_0, T^{\full}_1, N^{\full}_0, N^{\full}_1; A, C^{\full}_0, C^{\full}_1) = \Bigl( L, A, \Delta=I(T^{\full}_A \leq C^{\full}_A), X=\min\{T^{\full}_A, C^{\full}_A\}, N(\cdot) = N^*_A(\min\{\cdot, C^*_A\} ) \Bigr),$  where the coarsening variables determine which part of the full data is available.  Define $\P$ as the distribution of the observed data induced by $\Pf$. The observed recurrent event process $N(\cdot) = N^*_A(\min\{\cdot, C^*_A\})$ is possibly terminated early by censoring. Also define the observed failure process $N_T(t) = I(X \leq t, \Delta=1)$ and the observed censoring process $N_C(t) = I(X \leq t, \Delta=0)$. Estimation and inference, described later, 
will assume the availability of $n$ i.i.d.  copies of $\mO.$ 

It is important to remember that the recurrent event processes $N^{**}_a, N^*_a$, and $N$ are different; in particular, each process counts possibly different jumps. The hypothetical process $N^{**}_a$ is the starting point and jumps at all (potential) recurrent events. The critically important process that will later be used to define our desired causal estimand, or $N^*_a,$ jumps at all events that occur before, but not after, failure $T^*_a$. The process $N$ is observable and jumps at all events that occur while a unit is under observation (i.e., remains at risk). 

Above, and throughout, we use an asterisk over functions and random variables to denote reliance on the full or coarsened variables; a function or random variable without an asterisk only depends on the observed data and its associated distribution. 

Define $\vee$ as the maximum operator and $\wedge$ as the minimum operator, so that $a \vee b = \max\{a,b\}$ and $a \wedge b = \min\{a,b\}$. 
Define the overbar notation as denoting the history of a process, for example, $\bar{N}^*_a(t) = \{ N^*_a(u) \, : \, 0 \leq u \leq t \},$ where
$N^*_a(u) = 0, a=0,1.$ Define $\prodi$ as the product integral \citep{gill1990survey}. 
Finally, we adopt the commonly used convention that $0/0 = 0$. 

\section{The Estimand}
\label{sec:estimand}

Define $\mu_a^{\full}(t) = \E^{\full} \{ N_a^{\full}(t) \}$ as the expected number of counterfactual recurring events. The value of $\mu_a^{\full}(t)$ is determined by two factors and will be small if  (1) failure is likely to occur early, or (2) if recurrent events are rare.  In the former case, the recurrent event count $N^*_a(t) = N^*_a(T^*_a \wedge t)$ will be small since the counting process is stopped early, while in the latter case, the recurrent event process itself is a driving factor. To assess the contribution of early failure, define $\eta^*_a(t) = \P^{\full} ( T^*_a > t )$, the  counterfactual survival probability.  This quantity remains central due to its direct influence on $\mu^*_a$. 


Define $\psi(\Pf)=\Bigl( \mu^*_0(t_1), \mu^*_1(t_1), \eta^*_0(t_1), \eta^*_1(t_1),    \dots,     \mu^*_0(t_m), \mu^*_1(t_m), \eta^*_0(t_m), \eta^*_1(t_m) \Bigr)$ as  the full data estimand, where  $0  < t_1 < \dots < t_m$ is pre-specified  sequence of $m$ landmark times. These $4m$ values give the functions of interest at each landmark time and in each treatment arm. Much of the development that follows considers estimation of $\mu^*_a(t)$ and $\eta^*_a(t)$ at $t = t_j$ for a given $j > 0$ and $a=0,1,$ and will be sufficient for the study of $\psi(\Pf).$

\subsection{Identifiability conditions}
\label{sec:estimand:car}

We impose the following identifiability conditions to link the full and observed data.

\begin{assp} \label{prop:car}
   Let $N^{\full}_a(\cdot)$ denotes the full path of the counterfactual recurrent event process $N^{\full}_a$ We impose the following identifiability conditions: (i) $(N^{\full}_0(\cdot), N^{\full}_1(\cdot), T^{\full}_0, T^{\full}_1) \ind A \mid L,$ (ii) $(N^{\full}_0(\cdot), N^{\full}_1(\cdot), T^{\full}_0, T^{\full}_1) \ind C^{\full}_a \mid A=a, L \text{ on } C^{\full}_a < T^{\full}_a,$ (iii) $(N^{\full}_0(\cdot), N^{\full}_1(\cdot), T_{1-a}^{\full}) \ind I(T_a^{\full} \leq C_a^{\full}) \mid A=a, L, T^{\full}_a$. \end{assp}

The first condition is that individuals with the same profile (i.e., with the same baseline variables $L$) have their exposure $A$ independent of their potential outcomes. It is a version of the ``no unmeasured confounders'' or ``strong ignorability'' assumption \citep{rosenbaum1983central}.  The second and third conditions follow from the classical assumption of conditionally independent censoring, namely, $(N^{\full}_0(\cdot), N^{\full}_1(\cdot), T^{\full}_0, T^{\full}_1) \ind C^{\full}_a \mid A=a, L$, for each $a=0,1$. The second condition restricts this independence to the subset of times where censoring precedes failure \citep{van2003unified}. 

The characterization of the identifiability conditions in Assumption \ref{prop:car} leads to a \emph{sequential} identifiability condition, where the observed data is formulated as equaling the full data after multiple stages of coarsening. In our case, we formulate the causal selection (i.e., involving $A$) as occurring before censoring (i.e., involving $C^*_A$). 
\begin{assp}
\label{cor:seq-car}
    Let $N^{\full}_a(\cdot)$ denote the full path of the counterfactual recurrent event process $N^{\full}_a$. The sequential identifiability conditions are (i) $( N^{\full}_0(\cdot), N^{\full}_1(\cdot), T^{\full}_0, T^{\full}_1 ) \ind A \mid L,$ (ii) $(N^{\full}_a(\cdot), T^{\full}_a) \ind C^{\full}_a \mid A=a, L \text{ on } C^{\full}_a < T^{\full}_a$.
\end{assp}
 The conditional independence condition imposed in Assumption \ref{cor:seq-car} is weaker than those appearing in Assumption \ref{prop:car} and will be used for identification of the estimand. Throughout the remainder of this paper, we also impose  certain regularity and positivity assumptions, such as those listed in Section \ref{sup:reg:ass} of the Supplemental Appendix. 

\subsection{Identification via inverse probability weighting mapping}
\label{sec:estimand:identif:results}

Define the usual at-risk process $Y(u) = I(X \geq u)$. 
Define $H(t; a, l)=\prodi_{u\in (0,t]} \{ 1 - \mathrm{d} \Lambda_T(u; a, l) \}$ for $\Lambda_T(t; a, l) = \int_{(0, t]} \frac{\d \E \{ N_T(u) \mid A=a, L=l\}}{\E \{ Y(u) \mid A=a, L=l\}},$ where $\prodi$ is the product integral \citep{gill1990survey}. Similarly define $K(t; a, l) = \prodi_{u\in (0,t]} \{ 1 - \mathrm{d} \Lambda_C(u; a, l) \}$ for $\Lambda_C(t; a, l)     = \int_{(0, t]} \frac{\d \E \{ N_C(u) \mid A=a, L=l \}}{\E\{ Y^{\dagger}(u) \mid A=a, L=l\}},$
where $Y^{\dagger}(u) = I(X > u, \Delta=1 \text{ or } X \geq u, \Delta=0).$ The process $Y^{\dagger}(u)$ is 
a modified version of the usual at-risk process $Y(u) = I( X \geq u).$ In fact, $Y^{\dagger}(t) \, = \, Y(t) - 
\{ N_T(t) - N_T(t-) \},$ and is appropriate to use when events for which $\Delta=1$ have priority over events for which $\Delta=0$; see \citet[Page 56]{gill1994lectures} for related developments. For example, this is the case for the observed data since $\Delta_A = I(T^{\full}_A \leq C^{\full}_A)$ is defined with a ``$\leq$'' rather than a ``$<$''. Define $M_C(t;a,l)  := N_C(t) - \int_{(0,t]} Y^{\dagger}(u) \,\d \Lambda_C(u;a,l).$ 
The process $M_C$ plays an important role in each augmentation term in our setting; under the true $\P,$ it vanishes in expectation due to the appearance of $Y^{\dagger}$ in both the numerator and denominator.

\begin{remark}
Although the at-risk process $Y$ is left-continuous and predictable with respect to the observed data filtration \citep[e.g.,][]{fleming1991counting},
the modified process $Y^{\dagger}$ is neither left-continuous nor right-continuous, and fails to be predictable. Nevertheless, the process $M_C(t;a,l)$ 		remains a martingale process with respect to the same filtration; see \citet{baer2024theory} for these and other related results.
\end{remark}

Let $\Pf$ be any full and coarsened data distribution with corresponding observed data distribution $\P \in \mathcal{M}$ where $\mathcal{M}$ is the model space induced by the identification conditions imposed earlier. \citet{baer2025survival} showed that $K, H$ respectively identify $K^{\full}, H^{\full}.$ This result only relies on Assumption \ref{cor:seq-car}, and does not require the more typical (and more restrictive) assumption of full conditional independence between the potential censoring and failure times. Additionally, the identification result relies on the modified at-risk process $Y^{\dagger}$ for censoring, and avoids the need to impose separate absolute continuity or discreteness assumptions on either failure or censoring
\citep[e.g.,][]{scharfstein2001inference, scharfstein2002estimation}. Theorem \ref{prop:identif} shows that under our previously stated assumptions, the desired full data estimand is identified in $\M$, as the following result shows. 
\begin{thm}
\label{prop:identif}
    Let $\Pf$ be any full and coarsened data distribution with corresponding observed data distribution $\P \in \mathcal{M}$. For each landmark time $t>0$,  the following statements holds: (i) The full data estimand component $\mu_a^{\full}(t)$ equals $\mu_a(t)$, where $\mu_a(t) = \E \{ \varphi_{\mu, a} (t; \mO; \P) \}$ for
        $\varphi_{\mu, a} (t; \mO; \P)
        = \frac{I(A=a)}{\pi(A; L)} \frac{\Delta}{K(X-; A, L)} N(t).$ (ii) The full data estimand component $\eta_a^{\full}(t)$ equals $\eta_a(t)$, where $\eta_a(t) = \E \{ \varphi_{\eta, a} (t; \mO; \P) \}$ for
        $\varphi_{\eta, a} (t; \mO; \P)
        = \frac{I(A=a)}{\pi(A; L)} \frac{\Delta}{K(X-; A, L)} I(X > t).$ (iii) For any $u>0$, the parameter $F^*_a(u, t; a, l) = 
        \E^\full\{ I(T^\full_a > u) N^\full_a(t) | A = a, L = l \}$ equals $F(u, t; a, l),$ where 
        $F(u, t; a, l) = \E \left\{ \frac{\Delta}{K(X-; A, L)} I(X>u) N(t) \,\middle|\, A=a, L=l \right\}.$
\end{thm}   
The function $F(\cdot, t; a, l)$ will play an important role in our proposed  estimator.

\subsection{A von Mises expansion for $\mu_a$
and $\eta_a$}
\label{sec:estimand:vonmises}

Define $\P$ and $\bP$ as two observed data distributions. Write $\E$ and $\bE$ as the expectation with respect to these distributions, respectively. 
For a given parameter $\psi$, a von Mises expansion of $\psi$ at $\bP$ centered at $\P$ is $\psi(\bP) - \psi(\P) = (\bE - \E) D(\mO; \bP) + R(\bP, \P),$ where $D$ captures the first order behavior of $\psi$ and $R$ captures the remainder. Below, we state this decomposition for each component of the observed data estimand.
\begin{thm}
\label{prop:vonmises}
 Consider two observed data probability distributions $\P$ and $\bP$. For each time $t>0$, the following statements hold for the observed data estimands $\mu_a(t)$
 and $\eta_a(t):$
 \begin{itemize}
     \item[i.] The estimand $\mu_a(t)$ admits a von Mises expansion with 
     $D_{\mu, a}(t, \mO; \P)$ equal to 
    \begin{align*}
        \varphi_{\mu, a} (t; \mO; \P) - \mu_a(t) 
        - \frac{I(A=a) - \pi(a; L)}{\pi(a; L)} F(0, t; a, L) 
        + \frac{I(A=a)}{\pi(a; L)} \int_{(0, \infty)} \frac{F(u, t; a, L)}{H(u; a, L)} \frac{\d M_C(u; a, L)}{K(u; a, L)}. 
    \end{align*}
   \item[ii.] The estimand $\eta_a(t)$ admits a von Mises expansion with $D_{\eta, a}(t, \mO; \P)$ equal to 
    \begin{align*}
        \varphi_{\eta, a}(t, \mO; \P) - \eta_a(t)
            - \frac{I(A=a) - \pi(a; L)}{\pi(a; L)} H(t; a, L) 
            + \frac{I(A=a)}{\pi(a; L)} \int_{u\in(0,\infty)} \frac{H(t \vee u; a, L)}{H(u; a, L)} \frac{\d M_C(u; a, L)}{K(u; a, L)}.
    \end{align*}
 \end{itemize}
 The explicit form of the remainder terms $R_{\mu, a}(t, \mO; \bar{\P}, \P)$ and $R_{\eta, a}(t, \mO; \bar{\P}, \P)$ are presented in 
Section \ref {sec:supp:vm:cont} of the Supplemental Appendix, where each is also shown to be of second order 
  \citep[e.g.,][]{kennedy2023semi}.  
\end{thm}
A von Mises expansion for the full estimand can be readily constructed by stacking the components in Theorem \ref{prop:identif}. 
The existence of this expansion indicates that the observed data estimand is sufficiently smooth to be pathwise differentiable \citep{bickel1993efficient, van2011targeted}, guaranteeing the existence of at least one asymptotically linear estimator.  
Under a nonparametric model with no restrictions on $\P$, the von Mises expansion  of $\psi$ at $\bP$ centered at $\P,$ or $\psi(\bP) - \psi(\P)
= (\bE - \E) D(\mO; \bP) + R(\bP, \P),$ holds for any $\bP$ and $\P,$ the  corresponding ``derivative'' $D(\mO; \bP)$ is also the efficient influence function.  However, in the case where a model is semiparametric, with $\P$ being subject to additional restrictions, there can potentially be many influence functions, each corresponding to a different von Mises expansion \citep[e.g.,][]{kennedy2023semi}. The identifiability conditions in Assumptions 1 and 2 create certain restrictions on $\P,$ and the expansions in Theorem \ref{prop:vonmises} respectively represent one of potentially many possible choices that may or may not be (locally) efficient.

\section{Estimation and Inference}
\label{sec:estimation}

Consider a model $\M_n$ containing a sample of $n$ independent copies from a distribution in $\M,$
where $\M$ is described earlier. Let $\P_0 \in \M_n$ denote the true observed data distribution satisfying  Assumption \ref{cor:seq-car}. Below, functionals subscripted by $0$ are to be evaluated at $\P_0$. When the observed data $\mO$ is sampled from  $\P_0,$ an estimator $\hat{\psi}_n$ for an observed data estimand $\psi_0 = \psi(\P_0)$ has influence function $\mathrm{IF}(\mO; \P_0)$ provided that  $\hat{\psi}_n - \psi_0   = \E_n \left\{ \mathrm{IF}(\mO; \P_0) \right\} + o_{\P_0}(n^{-1/2}).$ Up to an asymptotically negligible term $o_{\P_0}(n^{-1/2})$, this expansion expresses that the estimation error $\hat{\psi}_n - \psi_0$ is a sample average of i.i.d. terms $\mathrm{IF}(\mO; \P_0)$. 

\subsection{The class of influence functions}
\label{sec:estimation:class}

Define the model $\M_n(\pi_0,K_0)$ as the restriction of $\M_n$ when the propensity score $\pi_0$ and the censoring survival function $K^*_0=K_0$ are known and hence fixed. The model $\M_n(\pi_0,K_0)$ may arise in a randomized trial with perfect adherence and a fully known
mechanism for censoring. Define $\varphi(\mO; \P_0) \in \mathbb{R}^{4m}$ as a stacking of the expressions in Theorem \ref{prop:vonmises} for each estimand component. For example, the component of $\varphi$ corresponding to $\mu_0(t_1)$ is $\varphi_{\mu, 0}(t_1, \mO; \P_0)$. Below,
we give the class of influence functions for $\psi_0$ in $\M_n(\pi_0,K_0).$
\begin{thm}
\label{prop:class}
    For each component of the estimand, the class of influence functions for $\psi_0$ in $\M_n(\pi_0,K_0)$ has typical element
    \begin{align*}
      & \varphi_a (\mO; \P_0)
      - \psi_0
      - \Bigl\{ I(A=a) - \pi_0(a; L) \Bigr\} h_1(L) + \int_{(0,\infty)} h_2 \{ u; \bar{N}(u), A, L \} \,\d M_{C,0}(u; A, L),
    \end{align*}
    where $\varphi_a$ is the inverse probability weighted expression in Theorem \ref{prop:identif} (i.e., when
    $\P = \P_0$) and $h_1$ and $h_2$ are arbitrary index functions.
\end{thm}

 Formally, a typical element is defined to be in the closure of the finite variance elements in the above class \citep{bickel1993efficient}. 
The class has an augmentation term due to $\pi_0$ and another augmentation term due to $K_0$. The augmentation space is defined to be the class of influence functions with the first summand $\varphi_a - \psi_0$ removed.

The influence functions of all regular and asymptotically linear estimators in $\M_n(\pi_0, K_0)$ lie in this class.  We may interpret this class of influence functions as a corresponding class of estimating functions, since estimators obtained by solving these estimating equations have influence functions that also belong to this class. For example, when $h_1, h_2=0$, we see that the estimating functions for inverse probability weighted estimators lie in the class. 

The class of influence functions in $\M_n(\pi_0,K_0)$ is important even when $\pi_0$ and $K_0$ are not known. When the coarsening probabilities are known to lie in a smooth parametric model, \citet[Theorem 8.3]{tsiatis2006semiparametric} shows that the class of influence functions is a projection of the above class onto the coarsening tangent space. 

The proof of Theorem \ref{prop:class} utilizes the following lemma which may be of independent interest. The lemma shows that asymptotic theory does not rely on whether the coarsening variables are defined as $A, C^*_0, C^*_1$ or simply as $A,C^{\full}$. 
This is because the observed data identification of the conditional censoring distribution
(i.e., see \S \ref{sec:estimand:identif:results}) cannot distinguish between the coarsening variables
$C^\full$ and $C^\full_A = A C^\full_1 + (1-A) C^\full_0;$ as a consequence, the augmentation space for estimating the distribution of $T^\full_a$ under each censoring model is identical.

Define the censoring survival function $K_{\text{single}, 0}(t; a, l)$ as the identification of $K^{\full}_{\text{single}, 0}(t; a, l)= \PP_0(C^{\full} > t \mid A=a, L=l)$. Define $\M_{\text{single}, n}$ as the model corresponding to $\M_n$ but instead defined with coarsening variables $A,C^{\full}$. 
More information on the definition of $\M_{\text{single}, n}$ is given in Section \ref{sec:supp:if-class} of the Supplemental Appendix.

\begin{lem}
\label{lem:single-reduction}
 Define $\mathrm{AS}_{2,\text{single}}(\pi_0, K_{\text{single}, 0})$ as the augmentation space in $\M_{\text{single}, n}(\pi_0, K_{\text{single}, 0})$, which depends on the coarsening probabilities $\pi_0$ and $K_{\text{single}, 0}$. Then the augmentation space in $\M_n(\pi_0, K_0)$ is $\mathrm{AS}_{2,\text{single}, 0}(\pi_0, K_0)$.
\end{lem}

\subsection{The  influence function}
\label{sec:estimation:eif}

In this section, we characterize an influence function corresponding to our multiply robust estimator in $\M_n$. Define $D(\mO; \P_0) \in \mathbb{R}^{4m}$ as a stacking of the first order terms of the von Mises expansion in Theorem \ref{prop:vonmises} for each estimand component. For example, the component of $D$ corresponding to $\mu_0(t_1)$ is $D_{\mu, 0}(t_1, \mO; \P_0)$. 
\begin{thm}
\label{thm:eif}
 The  influence function for estimating $\psi_0$ in the model $\M_n,$ or in the submodel $\mathcal{M}_n(\pi_0, K_0)$ where $\pi_0$ and $K_0$ are known, is $D(\mathcal{O}; \mathbb{P}_0)$ as defined in Theorem \ref{prop:vonmises}.
\end{thm}
\begin{remark}
The  influence function for $\eta_a(t)$ given as part of Theorem \ref{thm:eif} is pointwise
identical to that given in \citet{westling2021inference}; see \citet{baer2025survival} for details,
and Section \ref{sec:conc:eff} for further discussion.
\end{remark}

\subsection{Multiply robust estimators}
\label{sec:estimation:master}

In this section, we define an estimator for $\psi_0$ based on the derived influence function 
and study its asymptotic behavior. Reparameterize the influence function as $D(\psi_0, \mO; \theta_0)$, where $\psi_0$ is 
the estimand and $\theta_0 = (\pi_0, K_0, H_0, F_0)$ is a vector comprised by the nuisance 
parameters. Define the initial estimator $\hat\psi_{n}$ as solving the 
influence function as an estimating function; under mild conditions, the linearity of $D$ implies
that $\E_n \left\{ D(\hat\psi_{n}, \mO; \hat{\theta}_n ) \right\}
    = o_{\P_0}(n^{-1/2})$
for some nuisance parameter estimator $\hat\theta_n$. In Section \ref{sec:numeric}, 
we briefly outline one approach to estimate the nuisance 
parameter $\theta_0$, denoting the estimator as $\hat{\theta}_n$.  We require  any 
estimator be \emph{cross-fitted} 
\citep[e.g.,][]{chernozhukov2018double};
specifically, we adopt the “DML2” version described in \citet{chernozhukov2018double}. 
For notational simplicity, we omit explicit cross-fitting notation throughout the paper.
The following result provides conditions under which $\hat\psi_{n}$ is asymptotically linear with influence function $D(\mO; \P_0).$
\begin{thm}
\label{thm:est-master}
    Let $\P_0 \in \M_n$. We impose the following assumptions: (1) the nuisance parameter estimators $\hat\theta_n$ are cross fit; and,
    all $n$ sufficiently large, we have
        (2) $\epsilon' < \pi_n(a; l) K_n(\tau; a, l)$ for some $\epsilon'>0$ and (3)
         $F_n(u, t; a, L) \leq C_F' H_n(u; a, L)$ for all $u, t > 0$ and $a=0,1$, almost surely, for some $C_F' < \infty$.
    Under i.i.d.\ sampling, conditions (1)-(3) given above, Assumption \ref{cor:seq-car}, consistency of all nuisance parameter
    estimators, and the assumptions given 
    in Section \ref{sup:reg:ass} of the Supplemental Appendix,
    the estimator $\hat{\psi}_n$ satisfies the expansion
    $ \hat{\psi}_n - \psi_0 
        = \E_n \left\{ D(\mO; \P_0) \right\}
        + o_{\P_0}(n^{-1/2}) 
        + O_{\P_0}(r_n),$
    where $D(\mO, \P_0)$ is an influence function in $\M_n$ and the explicit form of $O_{\P_0}(r_n)$ is given in Section \ref{sec:supp:if-class} of the Supplemental Appendix. It follows that $D(\mO; \P_0)$ is the influence function of $\hat\psi_n$ if $O_{\P_0}(r_n) = o_{\P_0}(n^{-1/2})$. 
\end{thm}

The assumptions of the theorem are mild.  The first simply states that the nuisance parameter estimators are cross fit. 
The second and third are the empirical analog of Assumptions \ref{asp:positivity} and \ref{asp:recur-event-growth} in 
Section \ref{sec:supp:if-class} of the Supplemental Appendix, respectively ensuring positivity and that $H_n = 0$ implies that $F_n = 0$ so that the ratio $F_n/H_n$ is well defined. The final assumption that each remainder is $o_{\P_0}(n^{-1/2})$ requires certain smoothness, sparsity, or c\'adl\'ag with finite sectional variation norm \citep{bibaut2019fast}. 

Considering the remainder derived in Theorem \ref{prop:vonmises}, our proposed estimators for both $\mu_a(t)$ and $\eta_a(t)$ will generally be consistent provided that either $(F_n,H_n)$ or $(\pi_n,K_n)$ consistently estimate their population counter-parts. The proof of Theorem \ref{thm:est-master} gives bounds on the remainder term that also suggest $\hat\psi_n$ has certain multiple robustness properties, provided  consistent nonparametric estimators are used for all nuisance parameters. For example, and in the extreme, $(\pi_n,K_n)$ converging at the usual parametric rate allows $(F_n, H_n)$ to converge arbitrarily slowly; more generally, faster rates of convergence for $(\pi_n,K_n)$ allows $(F_n,H_n)$ to have slower convergence rates while still guaranteeing that the remainder is $o_{\P_0}(n^{-1/2}).$ These same statements hold when reversing the roles of $(\pi_n,K_n)$ and $(F_n, H_n)$. 

Importantly, the estimator $\hat\psi_{n}$ may violate global constraints like monotonicity. To address this, we can define  $\hat\psi_{n, \mathrm{proj}}$ as its $L_2(\P_n)$ projection onto the space of increasing functions via isotonic regression. Under sufficient regularity, the results of \citet{westling2020correcting} imply that $\hat\psi_{n, \mathrm{proj}}$ is asymptotically equivalent to $\hat\psi_n$, with the potential for finite-sample improvement from enforcing monotonicity.

We note that the above results could also be established without cross-fitting of the nuisance parameters, provided that the true and estimated nuisance parameters were known to lie in a Donsker class, hence not ``too complex'' as functions of their arguments.

\section{Causal Estimands for Recurrent Events Stopped by Failure: Prior Work}
\label{sec:litreview}
\label{sec:litreview:recur}

In this section, we review the literature for estimating $\mu^*_a= \E^*\{ N^*_a (t) \}$ and thereby contextualize our results. We note
that in non-causal settings, the expected number of recurrent events before failure was first studied by \citet{cook1997marginal}. Additional studies include those by \citet{ghosh2000nonparametric,strawderman2000estimating}, among numerous others. The estimand $\mu^*_a$ reduces to such a non-causal estimand when $A=a$ almost surely, as this restriction  eliminates the propensity score component from the  influence function.

\subsection{Interpretability}
\label{sec:litreview:recur:interp}

Most work on this problem, both in lifetime analysis and in causal lifetime analysis, either notes or stresses that $\mu^*_a$ entwines the behavior of the (pure) recurrent event process $N^{**}_a$ and of the failure time $T^*_a$. 
Since $\eta^*_a$ is readily estimated, it is widely understood that $\mu^*_a$ should be interpreted in the context of $\eta^*_a$. However, exactly how best to do so remains a matter of uncertainty and lacks consensus. \citet{janvin2022causal} show that separability of the exposure can help
to disentangle the behavior of $N^{**}_a$ and $T^*_a$, but their approach relies on strong untestable assumptions. 
In this work, we do not propose a formal procedure for helping to isolate the causal effect for the recurrent process. Instead, for a given value $\mu^*_a(t)$, we simply inspect the history of $\eta_a^*(u)$ for $0 < u \leq t$. Further analysis and discussion is provided in \ref{sec:conc:interp}. 

\subsection{Some existing estimators (causal case)}
\label{sec:litreview:recur:existing-est}

As noted previously, many observed-data approaches to estimate $\mu_a(t)$ have been proposed in the absence of causal inference considerations  (i.e. when $A=a$ almost surely). In a causal inference setting, the earliest work directly studying $\mu^*_a$ is (to the authors' knowledge) \citet{schaubel2010estimating}, where two estimators are proposed: the first imputes censoring times and weights by the propensity score, while the second is a ``double'' inverse probability weighted estimator that solves  $\P_n \left\{  \frac{I(A=a)}{\pi_0(a; L)} \int_{(0, t]} \frac{\d N(u)}{K_0(u-; a)} - \mu_a(t) \right\} = 0,$ with $K_0(\cdot;a)$ being estimated by treatment-specific Kaplan-Meier estimators.
Below, we show that the influence function of this latter estimator belongs to the class given in Theorem \ref{prop:class}.
\begin{prop}
    The oracle estimator of \citet{schaubel2010estimating} has influence function with indices 
        $h_1(L)  = 0$, $ 
        h_2(u; A, L, \bar{N}(u))  = 
        I(A=a) N(u\wedge t) / \left[ 
        \pi_0(a; L)K_0(u; a) \right].        $
\end{prop}
\noindent The proof relies on the fundamental identities in Section \ref{eq:rr3} of the Supplemental Appendix. When the nuisance parameters are efficiently estimated in a smooth parametric model, the influence function is modified through a projection \citep[Theorem 9.1]{tsiatis2006semiparametric}. 

\citet{janvin2022causal} propose a more general class of inverse probability weighted estimators in this same causal estimation problem that includes the double inverse probability weighted estimator of \citet{schaubel2010estimating} as a special case; in the case of time-independent $L,$ it can be similarly shown that the influence functions of their estimators for $\mu_a(t)$ and $\eta_a(t)$ fall into our class. In the absence of failure 
\citet{jensen2016marginal} propose a related semiparametric marginal structural intensity model for a longitudinal exposure, whereas
 \citet{su2020doubly} proposed an augmented estimator for $\mu_a(t),$ and assert that their estimator has a double robustness property that holds under the conditions $(C^*_0, C^*_1, N^*_0(\cdot), N^*_1(\cdot)) \ind A \mid L$ and $C^*_a \ind N^*_a(\cdot) \mid L$ for $a=0,1$. Importantly, however, these authors also make a critical but unstated assumption: $C^*_a \ind L \mid A=a$ for $a=0,1$. Under a violation of this relatively strong assumption, their inverse probability weighted estimator is inconsistent and their augmented estimator is generally not doubly robust. The need for this assumption widens the class of influence functions in their (implicit) model; consequently, their influence function does not generally lie in our class. Section \ref{sec:supp:litreview:steele} of the Supplemental Appendix contains further details, including a simple modification of their estimator that generates an influence function that falls into our class, and is doubly robust under certain conditions.

\section{Numerical Studies} 
\label{sec:numeric}

We conducted a numerical study to investigate the finite-sample performance of the proposed methods, specifically the estimation of each treatment-specific functional. For brevity, we summarize the study design, the details of which can be found in the Supplementary Materials. First, we generated the covariates as $L_1 \sim \text{Bernoulli}(0.5)$, $L_2 \sim \text{Unif}(-1, 1)$ and $L_2 \sim 0.5 + 3 \times \text{Beta(2,2)}$. Then the treatment assignment was generated according to a logistic regression model with $A \sim \text{Bernoulli}(\text{expit}(\beta_A^{\top}(1, L_1, L_2, L_3)^{\top}))$. We simulated the censoring time $C$ using the proportional hazards model
\begin{equation}
    \lambda_C(t;A,L) = \rho_{0,C}(t)\exp\{\beta_C^{\top}(1, A, L_1, L_2, L_3)^{\top}\},
\end{equation}
where the baseline hazard $\rho_{0,C}$ is of a Weibull distribution with scale and shape parameters both equal to 1. The coefficients for the treatment and censoring models,  $\beta_A$ and $\beta_C$ respectively, were specified according to three scenarios, to be discussed shortly. Finally, we simulated recurrent events and death in a similar manner to the censoring time, using proportional intensity/hazards models with not only main terms but also with some interactions among the treatment and the covariates. We truncated $T$ at the administrative censoring time $\tau = 12$, and recurrent events at the observed survival time $X = \min(T', C),$ where
$T' = T \wedge \tau$. The parameters for the recurrent events and death models remained the same in all three simulation scenarios, while parameters for the censoring and treatment assignment models varied. Overall, the simulation parameters were chosen such that $0.3 \le \mathbb{P}^*(A = 1 \mid L_1, L_2, L_3) \le 0.7$; $\mathbb{P}^*(C \ge 12) \ge 0.02$ and $\mathbb{P}^*(C \le T') \le 0.8$.

We used SuperLearner to estimate the propensity score $\pi(\cdot)$ and Survival SuperLearner \citep{westling2021inference} to estimate the conditional survival and censoring distributions $H(\cdot)$ and $K(\cdot)$. The libraries we used to estimate the propensity score, the conditional survival and censoring probability, and how we estimated the nuisance parameters $F(\cdot)$ are described in greater detail in 
Section \ref{sup:simulation} of the Supplemental Appendix. A 5-fold cross-fitting procedure was conducted to obtain the proposed {one-step AIPW} estimates. To evaluate the performance of our {one-step AIPW} estimators, we compared them with other possible estimators for $\mu_a^*(t)$ and $\eta_a^*(t)$. First, we considered IPW estimators for both $\mu_a^*(t)$ and $\eta_a^*(t)$, that is $\hat{\mu}_a^{\text{IPW}}(t)  = \mathbb{E}_n\left\{\frac{I(A=a)}{\pi(A;L)}\frac{\Delta}{K(X-;A,L)}N(t)\right\},$ and $\hat{\eta}_a^{\text{IPW}}(t)  = \mathbb{E}_n\left\{\frac{I(A=a)}{\pi(A;L)}\frac{\Delta}{K(X-;A,L)}I(X > t)\right\}$. 
  Because asymptotic linearity of IPW estimators requires nuisance parameters estimable at parametric rates, we used a logistic regression and a Cox proportional hazards model to respectively estimate the propensity score and censoring probabilities for $\hat{\mu}_a^{\text{IPW}}(t)$ and $\hat{\eta}_a^{\text{IPW}}(t)$. Another estimator for $\eta_a^*(t)$ is the one-step estimator proposed by \cite{westling2021inference}. As noted in the Remark following Theorem \ref{thm:eif}, it is asymptotically equivalent to our estimator because they share the same influence function. However, the finite-sample implementations differ, so we include their \texttt{R} implementation in our comparative study. And finally, we considered the double inverse weighting estimator for $\mu_a^*(t)$ proposed by \cite{schaubel2010estimating}, that is
$\hat{\mu}_a^{\text{SZ}}(t) = \mathbb{E}_n \left\{ \int_0^t \frac{I(A=a)}{\pi(A;L)}\frac{1}{K(u;A)}dN(u) \right\}$. Following \cite{schaubel2010estimating}, we estimated the propensity score model using logistic regression and the censoring distribution using a treatment-specific Nelson-Aalen estimator, that is
    $K_n^{\text{SZ}}(t;A) = \exp\{-\Lambda^{\text{SZ}}_{nC}(t;A)\} \hspace{5mm} \text{where} \hspace{5mm}
    \hat{\Lambda}^{\text{SZ}}_C(t;A)\} = \mathbb{E}_n \left\{\int_0^t\frac{ I(A = a)dN_C(u)}{\mathbb{E}_n\{I(A = a)I(X \ge u)\}}\right\}.$
There are two important differences between  $\hat{\mu}_a^{\text{SZ}}(t)$ and $\hat{\mu}_a^{\text{IPW}}(t)$. First, $\hat{\mu}_a^{\text{SZ}}(t)$ assumes that $C$ only depends on the treatment assignment $A$ but not the covariates $L$. Second, $\hat{\mu}_a^{\text{SZ}}(t)$ uses more information along the time interval $[0,t]$ while $\hat{\mu}_a^{\text{IPW}}(t)$ only uses the information at $t$.  

As discussed earlier, we investigated three scenarios for the propensity score and censoring models. In  Scenario 1, the propensity score model includes main effects for all covariates, and the censoring model depends only on $A$, as $\hat{\mu}_a^{\text{SZ}}(t)$ assumes. Scenario 2 differs from the first in that the censoring model depends not only on the treatment $A$, but also covariates $L_1, L_2, L_3$. Scenario 2 aims to simulate the situation where the assumption that censoring is independent of the covariates fails.  Scenario 3 keeps the same censoring model as the second, while the propensity score model only depends on $L_3$, a strong covariate for the recurrent event and death models. In the third scenario, the propensity score and the censoring models for the IPW estimators are estimated using $L_1$ and $L_2$ only. The simulation scenarios are summarized in Table~\ref{table:sim-scenario}. In each scenario, we generated 1000 datasets for sample sizes $n = 500, 1000, 1500, 2000, 2500$. 

\renewcommand*{\thefigure}{\arabic{figure}}
\renewcommand*{\thetable}{\arabic{table}}

\begin{table}[h!]
    \centering
    \begin{tabular}{c|ccc}
    \hline
     & Scenario 1 & Scenario 2 & Scenario 3 \\
    \hline
    IPW estimators & Correctly specified & Correctly specified & Misspecified \\
    SZ estimators & Correctly specified & Misspecified & Misspecified \\
    \hline \\
\end{tabular} 
\caption{Summaries of three scenarios in the simulation study.}
\label{table:sim-scenario}
\end{table}
The simulation results are summarized in Figures~\ref{fig:rmse-average} and \ref{fig:coverage-average},  giving the point-wise root mean squared error (RMSE) and coverage, respectively, of each estimator with respect to each sample size for
$t \in \{1, 2, 3, 4, 5, 6\}.$
Each point in the plots represents the average result over the 1000 simulation replications and over landmark times $t \in \{1, 2, 3, 4, 5, 6\}$. Separate results for each landmark time $t$ are presented in the Supplemental Appendix, and display similar trends shown in Figures~\ref{fig:rmse-average} and \ref{fig:coverage-average}. 
\begin{figure}[h!]
    \centering
    \includegraphics[width=\linewidth]{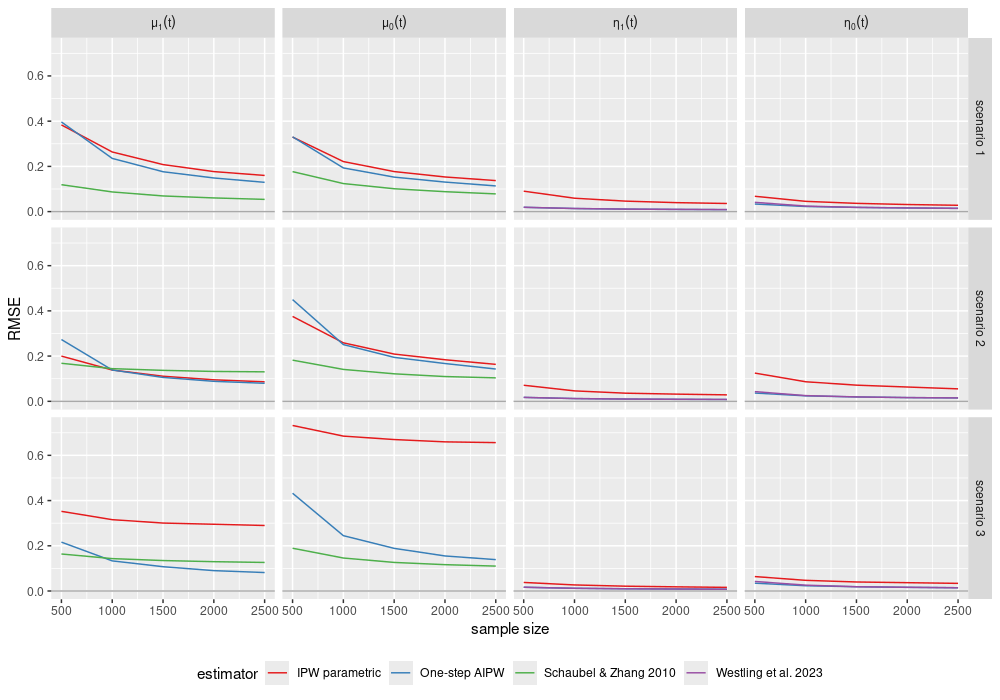}
    \caption{Point-wise root mean squared error (RMSE) of the estimators in three simulation scenarios as a function of the sample size $n$. Each point in the plot represents the average RMSE over estimators at $t \in \{1, 2, 3, 4, 5, 6\}$.}
    \label{fig:rmse-average}
\end{figure}
\begin{figure}[h!]
    \centering
    \includegraphics[width=\linewidth]{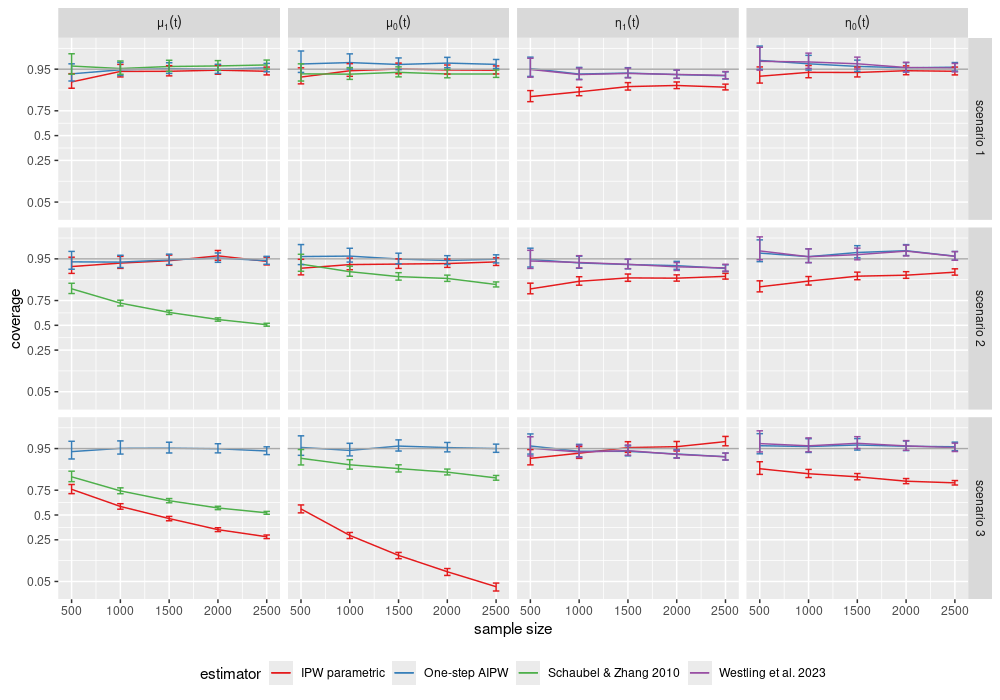}
    \caption{Coverage of the estimators for $\mu_1^*(t)$, $\mu_0^*(t)$, $\eta_1^*(t)$, $\eta_0^*(t)$ in three simulation scenarios as a function of the sample size $n$. Each point in the plot represents the average coverage over estimators at $t \in \{1, 2, 3, 4, 5, 6\}$. The vertical axis is plotted in the logistic scale, and the tick labels indicate values in the original scale. The error bars indicate 95\% confidence intervals considering uncertainty due to the finite number of simulation replications.}
    \label{fig:coverage-average}
\end{figure}

In the first scenario, where all estimators are correctly specified, we can see that all estimators for $\mu^*_a(t)$ provide correct coverage, and the RMSE decreases with increasing sample size. As $\hat{\mu}_a^{\text{SZ}}(t)$ uses more information over time than the IPW estimator, this estimator generally improves RMSE. Due to the greater uncertainty in nonparametrically estimating the nuisance parameter $F(\cdot)$, the finite-sample performance of our {one-step AIPW} estimators tends to lag behind $\hat{\mu}_a^{\text{SZ}}(t)$, with greater similarity in performance at larger sample sizes. As expected, we also observe that the {one-step AIPW} estimators gain efficiency compared to the IPW estimators. Similar trends present for estimators of $\eta^*_a(t)$, where the IPW estimators perform the worst, in terms of both coverage and RMSE.  Our estimates and those obtained using the \texttt{R} package by \cite{westling2021inference} are nearly identical, consistent with earlier
comments on their asymptotic equivalence. 

In the second scenario, where censoring probabilities depend not only on the treatment variable $A$ but also the covariates $L_1, L_2, L_3$, the performance of $\hat{\mu}_a^{\text{SZ}}(t)$ deteriorates as expected, with coverage worsening as sample size increases. Although the RMSE of $\hat{\mu}_a^{\text{SZ}}(t)$ remains low (i.e., mainly due to lower variance), the IPW and {one-step AIPW} counterparts tend to surpass it as sample size increases. This observation is more pronounced for $\mu_1^*(t)$ than for $\mu_0^*(t)$ because the models for the recurrent events and death respectively include interaction terms among $A$ and the covariates, making the $\mu_1^*(t)$ processes more dependent on the covariates than $\mu_0^*(t)$. The comparison of the IPW and {one-step AIPW} estimators in this scenario is 
similar to that in the first scenario.

The third scenario simulates the situation where the parametric nuisance models for IPW estimators are misspecified. As expected, the performance of the IPW estimators for both $\mu_a^*(t)$ and $\eta_a^*(t)$ in this case deteriorates in both coverage and RMSE, especially the coverage, which quickly decreases as sample size increases. The good performances of the proposed {one-step AIPW} estimators remain with correct coverages and low and decreasing RMSE. 

Section \ref{sup:simulation} of the Supplemental Appendix provides separate results for each landmark time and shows similar trends to those observed in Figures~\ref{fig:rmse-average} and \ref{fig:coverage-average}.

\section{Effect of PM$_{2.5}$ Exposure on Cardiovascular Disease Hospitalizations}
\label{sec:example}

PM${2.5}$ refers to airborne particles 2.5 micrometers or smaller that can penetrate deep into the lungs. Numerous studies have linked PM${2.5}$ exposure to adverse health outcomes, including respiratory illness, heart disease, and premature death \citep[e.g.,][]{xing2016impact,pun2017long}. 
Our analysis uses data from 272,226 Medicare beneficiaries who turned 65 between 2000 and 2016 while residing in a single ZIP code in Arizona. Each individual is followed for up to four years—two years of baseline exposure and up to two years of follow-up ($\tau = 24$ months), potentially truncated by death. Exposure and baseline covariates are based on ZIP-code-level averages during the baseline period. The treatment group ($A = 1$, $n =68060$) includes individuals in the highest quartile of PM$_{2.5}$ exposure ($\ge 9.11 \mu g/m^3$), and the control group ($A = 0$, $n = 68059$) includes those in the lowest quartile ($\le 4.75 \mu g/m^3$).

 We apply our method to estimate the point-wise counterfactual survival probabilities $\eta_a^*(t)$ and number of CVD-related hospitalizations $\mu^*_a(t)$ for every month after the 2-year baseline, that is, we consider estimates for landmark times $t = 1, 2, ..., 24$. Both $\pi(\cdot)$ and  $F(\cdot)$ are estimated as described in the numerical studies. The conditional survival and censoring functions are estimated using random forests for survival data with default hyperparameters specified in the \texttt{randomForestSRC} package in \texttt{R} \citep{ishwaran2007random}. We did not use survival SuperLearner \citep{westling2021inference} in this case due to the memory inefficiency of this package for large datasets. Our analysis results are shown in Figure~\ref{fig:AZ-mu-eta}. We can see that a high level of exposure to PM$_{2.5}$ increases the number of CVD-related hospitalizations, especially at 8 months and thereafter. The counterfactual survival probabilities are slightly reduced in the high exposure group, although for most of the landmark time points, the survival probabilities between the two groups are not statistically different. 

\begin{figure}[htbp!]
    \centering
    \includegraphics[width=\linewidth]{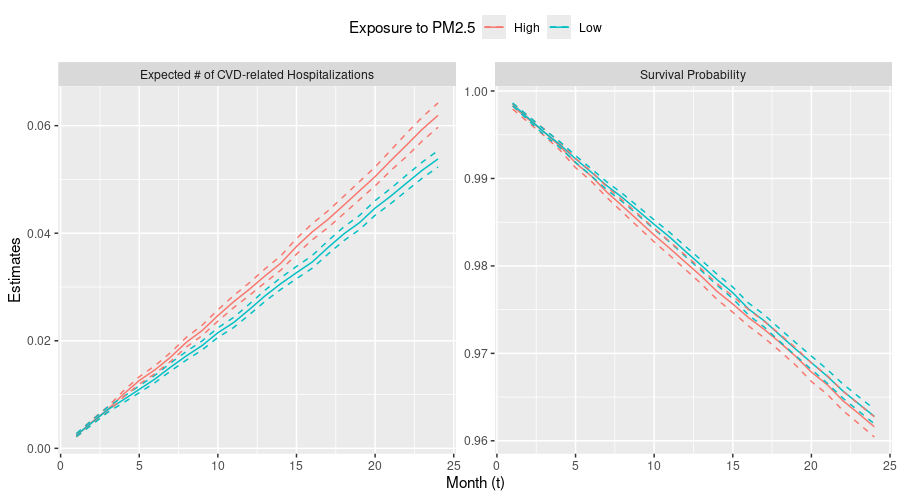}
    \caption{The left plot compares the estimates (full lines) and 95\% confidence intervals (dashed lines) of the expected number of CVD-related hospitalizations among groups of elder Medicare beneficiaries in Arizona, who have high or low exposure to PM$_{2.5}$. The right plot compares the estimates and 95\% confidence intervals of the survival probability among groups of elder Medicare beneficiaries in Arizona, who have high or low exposure to PM$_{2.5}$.}
    \label{fig:AZ-mu-eta}
\end{figure}

It is well known that the interpretation of $\mu_a^*(t)$ is complicated by the potentially competing dynamics of the recurrent event and death processes, and
that this is exacerbated when making comparisons between independent groups. For example, an increase in $\mu_1^*(t)$ compared to $\mu_0^*(t)$ might be entirely due to a corresponding increase in expected survival time in the treated group (i.e., we may expect more events simply because the subjects live longer). Although this is not the case in Figure~\ref{fig:AZ-mu-eta}, where the counterfactual survival curves appear to be very similar, it may still be helpful to consider a composite measure of effect, such as the ``while-alive strategy" causal estimand proposed in the literature \citep{schmidli2021estimands,maoWA2022,janvin2022causal}. This measure has specifically been proposed to address the interpretation problem of comparing $\mu^*_1(t)$ and $\mu^*_0(t).$ This causal estimand is defined as
\[ \frac{\mu_1^*(t)/\mathbb{E}^*[T_1^*(t)]}{\mu_0^*(t)/\mathbb{E}^*[T_0^*(t)]},\]
where $\mathbb{E}^*[T_a^*(t)]$ denotes the expected counterfactual restricted mean survival time up to $t$. The estimand can be interpreted as the ratio comparing ``How many hospitalizations I can expect, relative to how long I can expect to live in the next $t$ months'' between the two treatment groups. Fortunately, this estimand can be written as functions of $\mu_a^*(t)$ and $\eta_a^*(t)$ and hence can be derived from the estimates obtained from our proposed methods. The results in Figure~\ref{fig:AZ-causal} show that the ratio is significantly greater than 1 at $t \ge 5$, implying the increase in the number of CVD-related hospitalizations after adjusting for counterfactual restricted survival time among the high exposure group compared to the low exposure one.

\begin{figure}[htbp!]
    \centering
    \includegraphics[width=\linewidth]{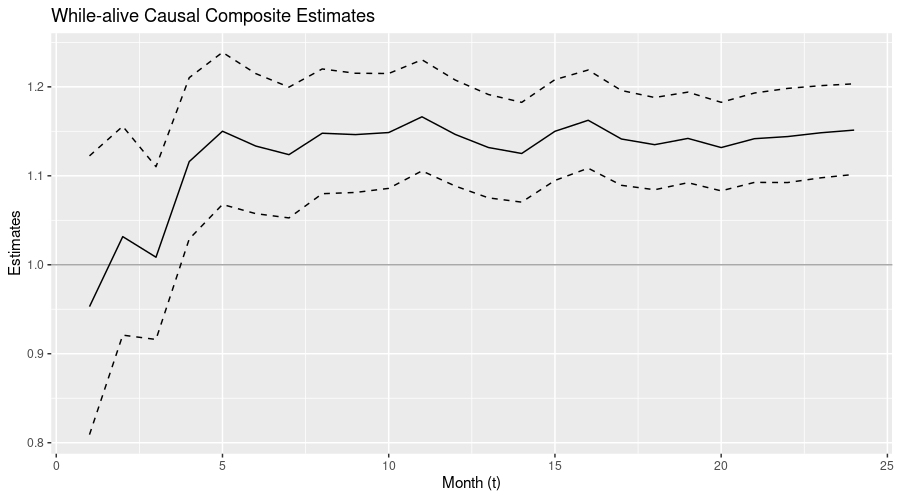}
    \caption{The estimates (full line) and 95\% confidence interval (dashed lines) for the while-alive causal composites estimand, i.e., the ratio of the expected number of CVD-related hospitalizations relative to the survival time between the high vs low exposure groups.}
    \label{fig:AZ-causal}
\end{figure}

Section \ref{sub-application} of the Supplemental Appendix contains additional results in which the treatment group is defined as being exposed to the top $35\%$ or $50\%$ of PM$_{2.5}$ levels and the control group is defined as being exposed to the bottom $35\%$ or $50\%$ of PM$_{2.5}$ levels, respectively. We can see that the difference between the treatment and control groups decreases as we narrow the difference in the PM$_{2.5}$ levels between the two groups. 

\section{Discussion}
\label{sec:conc}

\subsection{Efficiency considerations}
\label{sec:conc:eff}

One important advantage of our proposed estimator is the ability to nonparametrically estimate all nuisance parameters without needing to incorporate the continuous-time history of $N(\cdot),$ information that is not always readily available. This restriction comes at a disadvantage, 
for the corresponding model space $\mathcal{M}$ is smaller than the nonparametric model space induced under CAR (i.e., $\mathcal{M}^{CAR}$). This is reflected in our identifiability conditions,
in Section \ref{sec:estimand:car}, 
which are stronger than coarsening at random (CAR) since we do not condition on the history of the recurrent event process.   In view of the class of influence functions given in Theorem \ref{prop:class}, we conjecture that our estimator can be made efficient (i.e., within this smaller model) if we further assume $N_C(t) \ind N(t) \mid A,L$ and include the history of $N(t^-)$ (i.e., $\bar N(t^-)$) in the conditioning set of the nuisance functions $H$ and $F$ \citep[e.g.,][]{cort-sch2022}. In this case we define $H$ using a product integral of failure time intensities as $H(u) = \prodi_{s=0}^u \{1-d\Lambda(s \mid A,L,\bar N(t^-))\}$ where $d\Lambda(s\mid A,L,\bar N(t^-))$ is the conditional intensity corresponding to the failure time process; the estimators of \citet{baer2025survival} and \citet{westling2021inference} are no longer nonparametrically efficient in this case. A middle ground between the 
two extremes (i.e., conditioning on the full history versus none of it) is to make additional use of the process information at the landmark times only (i.e., while under observation). Efficiency might also be improved by modifying each of the IPW estimators used
to reflect the changing risk sets across landmark times.

\subsection{On causal interpretation}
\label{sec:conc:interp}

A formal causal interpretation of our primary estimand, $\mu^*_a(t)$, is nuanced, especially when considering its relationship with the survival function $\eta^*_a(t)$. While contrasts such as $\mu^*_1(t) - \mu^*_0(t)$ are often interpreted as total effects, they may reflect both direct and indirect effects through survival. Alternative approaches, such as the ``while alive'' strategy, attempt to jointly characterize recurrent events and survival, but still pose interpretational challenges under differential mortality. We propose a decomposition of $\mu^*_a(t)$ across landmark intervals to highlight the respective contributions of survival and recurrent event experience within each time interval. Although some terms in this decomposition lack strict causal interpretation, this can still offer useful descriptive insight into treatment effects. We refer the reader to Section \ref{supp:conc:interp} of the Supplemental
Appendix for further discussion and examples.

\subsection{On future work}

There is much work still be done in studying the full data estimand component $\mu^*_a$. Our approach relies on a sequential coarsening mechanism for identification, but exploring alternative methods such as instrumental variables, negative control variables, and other approaches would be valuable. The current work only considers a point-exposure setting; generalizing it to time-varying exposures is an important direction for future research. In settings with many time points,  methods based on ``under-smoothing'' may help mitigate the risk of obtaining irregular estimators with large biases due to the misspecification of some nuisance functions \citep{van2014targeted, ertefaie2023nonparametric}. Our proposed one-step estimator is derived from an estimating equation; however, other estimation strategies, such as targeted minimum loss estimators, may have improved finite sample performance \citep{van2011targeted}.  Finally, the results of this paper can be applied to estimate average treatment effects for recurrent event processes. While such an extension is of interest, a key challenge in the nonparametric framework is that conditional average treatment effects may not be pathwise differentiable, complicating inference. 



\section*{Acknowledgements}

BRB and AE were partially supported by the National Institute of Neurological Disorders and Stroke (R61/R33 NS120240). 
RLS, AE, DM, TB were partially supported by the National Institute for Environmental Health Sciences (R01 ES034021 [RLS, AE, DM, TB],
R01 ES035735 [DM]). 
AE was partially supported by the National Institute on Drug Abuse Health Sciences (R01 DA058996, R01 DA048764). 

\vspace*{-8pt}


%


\putbib[references]
\end{bibunit}

%% file: supp_doc.tex
\setcounter{page}{1}
\setcounter{section}{0}
\setcounter{equation}{0}
\setcounter{figure}{0}
\setcounter{table}{0}

\renewcommand*{\thesubsection}{S\arabic{section}.\arabic{subsection}}
\renewcommand*{\theequation}{S\arabic{equation}}
\renewcommand{\bibnumfmt}[1]{[S#1]}
\renewcommand{\citenumfont}[1]{S#1}

\renewcommand*{\thesection}{S\arabic{section}}

\begin{center}
    {\bf Supplemental Web Appendix} \\[2ex]
%
\end{center}


\begin{bibunit}[biom]

\section{Three Key Identities }
\label{supp:sec:fund-ident-proofs}

In this section, we state three identities that are important to deriving many of the results in this paper. Two of these are stated
and proved in \citet{baer2025survival}, along with some important background information; we restate those here without proof. A third
identity, important to the recurrent event problem considered in this paper, is proved as part of Lemma \ref{lem:rr-lems}.

Let $X \in (0, \tau]$ and $\Delta\in\{0,1\}$ be arbitrary numbers (or random variables) 
that need not be related in any way. Let $K$ be right-continuous, have locally bounded variation on any finite interval, and satisfy $K(0)=1$ and $K(u) \neq 0$ for all $u \in (0,\tau],$ where $\tau < \infty.$ Here, $K$ is an arbitrary function that is not necessarily the survival function of some potential censoring time. 
Below, we consider domains of integration as extending to $\infty;$ however, this is for convenience only, as each
stop at $\tau < \infty$ due to the assumptions imposed. 

Define
\begin{equation*}
\label{eq:lamc-defn}
    \Lambda_C(t) = - \int_{(0, t]} \frac{\d K(u)}{K(u-)};
\end{equation*}
note this is the cumulative hazard for censoring when $K$ is the censoring survival function. Define 
\begin{equation}
    M_C(t)
    = N_C(t) - \int_{(0,t]} Y^{\dagger}(u) \,\d \Lambda_C(u), \label{eq:mc}
\end{equation}
where  $Y^{\dagger}(u) = I(X > u, \Delta=1 \text{ or } X \geq u, \Delta=0).$

\begin{lem}
\label{lem:rr-lems}
Let $t>0$. 
Under the general setup just defined, the following identities hold:
\begin{align}
    \frac{I(X > t)}{K(t)} 
    & = \frac{\Delta}{K(X-)} I(X > t) + \int_{(t,\infty)} \frac{\d M_C(u)}{K(u)} \label{eq:rr1} \\
    \frac{\Delta}{K(X-)} 
    & = 1 - \int_{(0,\infty)} \frac{\d M_C(u)}{K(u)}. \label{eq:rr2}
\end{align}
Under the additional assumption that $N$ is right-continuous with $N(0)=0$ and $N(u) = N(u \wedge X)$ for any $u\geq 0$, the following also holds:
\begin{equation}
    \int_{(0,t]} \frac{\d N(u)}{K(u-)} 
    = \frac{\Delta}{K(X-)} N(t) + \int_{(0,\infty)} N(u\wedge t) \frac{\d M_C(u)}{K(u)}. \label{eq:rr3}
\end{equation}
\end{lem}

The proof of identities \eqref{eq:rr1} and \eqref{eq:rr2}
may be found in \citet{baer2025survival}; the proof
of \eqref{eq:rr3} is given below.
\begin{proof}[ Proof of \cref{lem:rr-lems} ]
 We begin by considering the term
    \begin{equation}
        \int_{(0,\infty)} N(u\wedge t) \frac{\d M_C(u)}{K(u)}
        = \int_{(0,t]} N(u) \frac{\d M_C(u)}{K(u)} + \int_{(t,\infty)} N(t) \frac{\d M_C(u)}{K(u)}. \label{eq:rr3-proof-aug}
    \end{equation}
    Applying the identity in \cref{eq:rr1}, the second summand in \cref{eq:rr3-proof-aug} may be rewritten as
    \begin{equation*}
        \int_{(t,\infty)} N(t) \frac{\d M_C(u)}{K(u)}
        = N(t) \left\{ \frac{I(X>t)}{K(t)} - \frac{\Delta}{K(X-)} I(X>t) \right\}.
    \end{equation*}
    The first summand in \cref{eq:rr3-proof-aug} may be directly simplified as
    \begin{align}
        \int_{(0,t]} N(u) \frac{\d M_C(u)}{K(u)}
        & = \int_{(0,t]} N(u) \frac{\d N_C(u) - Y^{\dagger}(u) \,\d \Lambda_C(u)}{K(u)} \nonumber \\
        & = N_C(t) \frac{N(X)}{K(X)} - \int_{(0,t]} N(u) Y^{\dagger}(u) \,\d \left\{ \frac{1}{K(u)} \right\}. \label{eq:rr3-aug2-proof}
    \end{align}
    Recalling that $Y(t) = I(X \geq t)$ defines the usual at-risk process, the last summand in \cref{eq:rr3-aug2-proof} is 
    \begin{align*}
        \int_{(0,t]} N(u) Y^{\dagger}(u) \,\d \left\{ \frac{1}{K(u)} \right\} 
        =\, & \int_{(0,t]} N(u) Y(u) \,\d \left\{ \frac{1}{K(u)} \right\}
         - \int_{(0,t]} N(u) I(X=u,\Delta=1) \,\d \left\{ \frac{1}{K(u)} \right\} \\
         =\, & \int_{(0, X \wedge t]} N(u) \,\d \left\{ \frac{1}{K(u)} \right\} - \Delta I(X\leq t) N(X) \left\{ \frac{1}{K(X)} - \frac{1}{K(X-)} \right\}
    \end{align*}
    In summary, we have shown that the augmentation term (in \cref{eq:rr3-proof-aug}) may be rewritten as
    \begin{align*}
        \int_{(0,\infty)} N(u\wedge t) \frac{\d M_C(u)}{K(u)}
        & = (1-\Delta)I(X\leq t) \frac{N(X)}{K(X)} \\
        & \hspace{-1in} - \int_{(0,X \wedge t]} N(u) \,\d \left\{ \frac{1}{K(u)} \right\} + \Delta I(X\leq t) N(X) \left\{ \frac{1}{K(X)} - \frac{1}{K(X-)} \right\} \\
        & \hspace{-1in} + N(t) \left\{ \frac{I(X>t)}{K(t)} - \frac{\Delta}{K(X-)} I(X>t) \right\}. 
    \end{align*}

    Before simplifying this expression, we apply integration by parts to show that
    \begin{align*}
        \int_{(0,t]} \frac{\d N(u)}{K(u-)} + \int_{(0, X \wedge t]} N(u) \,\d \left\{ \frac{1}{K(u)} \right\} 
        & = \int_{(0, X \wedge t]} \frac{\d N(u)}{K(u-)} + \int_{(0,t \wedge X]} N(u) \,\d \left\{ \frac{1}{K(u)} \right\} \\
        & = \frac{N(X \wedge t)}{K(X \wedge t)}.
    \end{align*}
    These expressions now simplify to the desired equality, noting again that $N(t) = N(X \wedge t)$.
\end{proof}

\section{Regularity assumptions}
\label{sup:reg:ass}

\begin{assp}
\label{asp:finitevar}
 All random variables have finite variance.
\end{assp}
\begin{assp}
\label{asp:phi}
 The map $\Phi$ correctly specifies the relationship between the full and observed data.
\end{assp}
\begin{assp}
\label{asp:trunc}
 There exists a cutoff $\tau < \infty$ such that $X \leq \tau$ almost surely and $t_m \leq \tau$.
\end{assp}
\begin{assp}
\label{asp:positivity}
 There exists $\epsilon>0$ so that $\epsilon < \pi_0(a; l) K^{\full}(\tau; a, l)$ almost surely for all supported $a,l$.
\end{assp}
\begin{assp}
\label{asp:recur-event-growth}
 There exists $C_F < \infty$ so that $F^*(u, t; a, L) \leq C_F H^*(u; a, L)$ for all $u, t > 0$ and $a = 0,1$ almost surely.
\end{assp}

\cref{asp:finitevar} in standard in developing asymptotically linear estimators \citep{bickel1993efficient}.  
\cref{asp:phi} captures the consistency assumption in causal inference that one of the counterfactuals is observed and the ``longitudinal assumption'' in survival analysis that the time at risk is the minimum of the potential failure and censoring times. 
\cref{asp:trunc} is a technical condition that truncates the whole real line and allows more accessible analysis \citep{gill1983large}. Together with \cref{asp:positivity}, it implies that the potential censoring distribution has support which exceeds the support of the potential failure distribution. This part of the assumptions can be forced to hold by redefining the potential failure time to be truncated at some sufficiently small $\tau>0$. 
\cref{asp:recur-event-growth} is a technical condition that ensures the recurrent event process does not explode; it holds, for example, when the recurrent process cannot jump more than $C_F$ times. 

The term $\pi_0(a; l) K^{\full}(\tau; a, l)$ appearing in \cref{asp:positivity} is the coarsening probability given the full data, that is, $\pi_0(a; l) K^{\full}(\tau; a, l) 
 = \Pf \bigl(A=a,C^{\full}_a>\tau \mid L=l, T^{\full}_0, T^{\full}_1, N^*_0(), N^*_1() \bigr).$
It is important to stress here that our causal estimation problem has more than one coarsening variable \citep[e.g.,][]{westling2021inference}.

\section{On causal interpretation}
\label{supp:conc:interp}

Overall, a universally applicable formalization of an estimand that interprets the primary estimand $\mu^*_a$ in the context of the secondary estimand $\eta^*_a$ is challenging and there is no satisfactory solution to date. For example, as noted in \citet{janvin2022causal},
the contrast estimand $\mu^*_1(t_k)-\mu^*_0(t_k)$ can be considered as representing a total effect. An important challenge in interpreting this measure of treatment effect is that a non-zero value is insufficient to 
establish that the applied treatment has a direct
effect on the mean number of recurrent events, independently of death; in particular, 
the total effect may additionally, or solely,
be a consequence of an indirect effect on
$T^{full}_a,$ as this can change the probability of remaining at risk. The ``while alive'' strategy of \citet{schmidli2021estimands},
used in \cref{sec:example}, instead tries to capture the treatment effect on both the recurrent event and failure outcomes; that is, the dynamic behavior of
\begin{equation} 
\label{eq:schmidli} \frac{\mu_1^*(t_j)/\mathbb{E}^*[T_1^*(t_j)]}{\mu_0^*(t_j)/\mathbb{E}^*[T_0^*(t_j)]}, ~j = 1,\ldots,m
\end{equation}
where $\mathbb{E}[T_a^*(t_j)] = \int_0^{t_j}\eta^*_a(u) du $ denotes the expected restricted counterfactual mean survival time up to $t_j.$ This latter estimand,
while preferred to the simple difference measure
$\mu^*_1(t_k)-\mu^*_0(t_k),$ requires care in its
interpretation when there is a differential
effect of treatment on mortality \citep{schmidli2021estimands}, hence suffers
from drawbacks related
to those that impair the utility of the total effect measure $\mu^*_1(t_k)-\mu^*_0(t_k)$ \citep{janvin2022causal}. 

Recall that our estimand is constructed for a fixed set of landmark times $0 = t_0 < t_1 < \dots < t_m$. Trivially, for a given landmark time $t_k,$
one may write $N(t_k) 
= \sum_{j=0}^{k-1} I\{ T > t_j \} \{  N(t_{j+1}) - N(t_j) \};$
hence, under our causal framework, we have the related decomposition
\begin{eqnarray}
    \mu^*_a(t_k) & = & 
\nonumber
    \sum_{j=0}^{k-1} \E^* \Bigl[ I\{ T^*_a > t_j \} ( N^*_a(t_{j+1}) - N^*_a(t_j) ) \Bigr] 
    \\
\label{eq:dis-comp}
   & = &
    \sum_{j=0}^{k-1} \eta^*_a(t_j) \Bigl\{ \zeta^*_a(t_{j+1}, t_j) - \zeta^*_a(t_j, t_j) \Bigr\},
\end{eqnarray}
where $\zeta^*_a(u,t) = \E^* \left\{ N^*_a(u) \mid T^*_a > t \right\}$ for $u \geq t$. Importantly, $\mu^*_a(t_k)$ is a valid causal estimand, and for similar reasons, so are $\theta^*_a(t_{j},t_{j+1}) = 
\E^* \Bigl[ I\{ T^*_a > t_j \} ( N^*_a(t_{j+1}) - N^*_a(t_j) ) \Bigr]$ and $\eta^*_a(t_j)$ for each $j.$ However, in the last decomposition and despite the fact that all terms are interpretable, $\zeta^*_a(t_r, t_j)$ is not a valid causal estimand for each $j\neq 0$ and $r \geq j,$ for it conditions on $T^*_a>t_j$, an event that differs between arms \citep{hernan2010hazards}. For similar reasons, the increment $\zeta^*_a(t_{j+1}, t_j) - \zeta^*_a(t_j, t_j)$ is also not a valid causal estimand. However, despite the lack of individual-level
causal interpretation, such a decomposition may
ultimately provide greater insight into the effects of treatment in comparison with a single composite measure like \eqref{eq:schmidli}. In particular, studying the pairs $\eta^*_a(t_j), \Bigl\{ \zeta^*_a(t_{j+1}, t_j) - \zeta^*_a(t_j, t_j) \Bigr\}$ allows one to 
separately evaluate the impact of surviving
to the beginning of each landmark interval
and the ensuing change in the mean count among
those still at risk. 

\begin{remark}
The decomposition \eqref{eq:dis-comp} is
exact, and is directly related to the 
familiar continuous decomposition
\begin{equation}
\label{eq:cont-decomp}
    \mu^*_a(t)
    = \int_{(0,t]} \eta^*_a(u) \zeta^*_a(\d u;u) 
\end{equation}
that has been used to motivate estimators
in a non-causal framework; see, for example, \citet{cook1997marginal,cook2009robust}.
\end{remark}

\section{Supplementary material for \cref{sec:estimand}}

\subsection{Observed Data Identification}

\begin{proof}[Proof of \cref{prop:identif}]
    The first part of the result follows through the calculation
    \begin{align*}
        \E \{ \phi_{\mu, a} (t; \mO; \P) \}
        & = \E \left\{ \frac{I(A=a)}{\pi(A; L)} \frac{\Delta}{K(X-; A, L)} N(t) \right\} \\
        & = \E^* \left\{ \frac{I(A=a)}{\pi(a; L)} \frac{I(T^*_a \leq C^*_a)}{K^*(T^*_a-; a, L)} N^*_a(t) \right\} \\
        & = \E^* \left\{ \frac{I(A=a)}{\pi(a; L)} \frac{N^*_a(t)}{K^*(T^*_a-; a, L)} \P^* \left( T^*_a \leq C^*_a \middle| A=a, L, T^*_a, N^*_a(t) \right) \right\} \\
        & = \E^* \left[ \frac{I(A=a)}{\pi(a; L)} \frac{N^*_a(t)}{K^*(T^*_a-; a, L)} \Bigl\{ 1 - \P^* \left( C^*_a < T^*_a \middle| A=a, L, T^*_a, N^*_a(t) \right) \Bigr\} \right] \\
        & = \E^* \left\{ \frac{I(A=a)}{\pi(a; L)} N^*_a(t) \right\} \\
        & = \E^* \left\{ \frac{N^*_a(t)}{\pi(a; L)} \P^* \left( A=a \middle| N^*_a(t), L \right) \right\} \\
        & = \E^* \left\{ N^*_a(t) \right\}.
    \end{align*}
    We remark here that the expression $\P^* \left( T^*_a \leq C^*_a \middle| A=a, L, T^*_a, N^*_a() \right)$ 
    needed to be handled delicately due to the censoring assumption that only furnishes independence ``on $C^*_a < T^*_a$''. The calculations for the second and third parts follow similarly; proof for the second
    part may be found in \citet{baer2025survival}.
\end{proof}

\subsection{Derivation of von Mises Expansion Remainder}
\label{sec:supp:vm:cont}

Based on the form of the efficient influence function for $\eta_a(t)$ in \citet{bai2013doubly}, we conjectured the form of the gradient for $\mu_a(t)$ and $\eta_a(t)$ for any data type (i.e. not necessarily absolutely continuous). Here, we prove that the conjectured gradient satisifies a von Mises expansion
\citep[e.g.,][]{kennedy2023semi} by showing that it vanishes in expectation and that the corresponding remainder is second order. 

\begin{proof}[Proof of \cref{prop:vonmises}]

We start by showing that the first order term
\begin{align*}
    D_{\mu, a}(t, \mO; \P) = &
    \varphi_{\mu, a}(t, \mO; \P) - \mu_a(t)
      - \frac{I(A=a) - \pi(a; L)}{\pi(a; L)} F(0, t; a, L) \\
      &
      + \frac{I(A=a)}{\pi(a; L)} \int_{(0, \infty)} \frac{F(u, t; a, L)}{H(u; a, L)} \frac{\d M_C(u; a, L)}{K(u; a, L)}
\end{align*}
vanishes in expectation. By \cref{prop:identif}, the first term is equal to the second term in expectation. The third term clearly vanishes in expectation upon applying iterated expectations. The expectation of the fourth and final term equals
\begin{align*}
    \E \biggl\{  \frac{I(A=a)}{\pi(a; L)} & \int_{(0, \infty)} \frac{F(u, t; a, L)}{H(u; a, L)} \frac{\d M_C(u; a, L)}{K(u; a, L)} \biggr\} \\
    =\, & \E \left[ \E \left\{ \int_{(0, \infty)} \frac{F(u, t; a, L)}{H(u; a, L)} \frac{\d M_C(u; a, L)}{K(u; a, L)} \middle| A=a, L \right\} \right];
 \end{align*}
given $A = a$ and $L,$ the inner integral is the expectation of a martingale process and hence
has mean zero \citep{baer2024theory}. Somewhat more informally, this can be seen by noting that
the right-hand side of the last displayed expresson can be written
\[
\E \left[ \int_{(0, \infty)} \frac{F(u, t; a, L)}{H(u; a, L)} \frac{\d \E \left\{ M_C(u; a, L) \middle| A=a, L \right\}}{K(u; a, L)} \right].
\]
Recalling \cref{eq:mc}, the inner conditional expectation equals
\begin{align*}
    \E \bigl\{ M_C(u; a, L) & \big| A=a, L \bigr\} 
    =\,  \E \left\{ N_C(u) - \int_{(0,u]} Y^{\dagger}(s) \,\d\Lambda_C(s; a, l) \middle| A=a, L \right\} \\
    =\, & \E \left\{ N_C(u) - \int_{(0,u]} Y^{\dagger}(s) \frac{\d \E \{ N_C(s) \mid A=a, L=l \}}{\E\{ Y^{\dagger}(s) \mid A=a, L=l\}}  \middle| A=a, L \right\} \\
    =\, & \E \left\{ N_C(u) \middle| A=a, L \right\} - \int_{(0,u]} \E \left\{ Y^{\dagger}(s) \middle| A=a, L \right\} \frac{\d \E \{ N_C(s) \mid A=a, L=l \}}{\E\{ Y^{\dagger}(s) \mid A=a, L=l\}} \\
    =\, & 0.
\end{align*}
Hence, $D_{\mu, a}(t, \mO; \P)$ vanishes in expectation; similarly,  $D_{\eta, a}(t, \mO; \P)$ vanishes in expectation.

We now show that the remainder term in the von Mises expansion is second order. Recall that the von Mises expansion  of $\psi$ at $\bP$ centered at $\P,$ or 
\[
\psi(\bP) - \psi(\P)
= (\bE - \E) D(\mO; \bP) + R(\bP, \P),
\]
holds for suitable $\bP$ and $\P.$ Each of $\mu_a$ and $\eta_a$ 
are examples of $\psi$, and each be expressed as the expectation of a data-dependent expression $\phi({\cal O}; \P).$ Hence, 
the remainder simplifies as
\begin{align*}
 \bE \phi({\cal O};\bP) - \E \phi({\cal O};\P) + \E D({\cal O};\bP)
 & = \E \phi({\cal O};\bP) - \E \phi({\cal O};\P) + \E D_{\text{aug}}({\cal O};\bP),
\end{align*}
where $D_{\text{aug}}$ is the augmentation term in $D$. 
Throughout the calculations below, we use the notation e.g. $\pi,K,H,F$ to denote functionals evaluated at $\P$ and $\bar\pi,\bar{K},\bar{H},\bar{F}$ to denote functionals evaluated at $\bP$.

We start by studying the remainder for $\mu_a(t)$. 
The first two summands simplify as
\begin{align}
 \E \phi({\cal O}; \bP) - \E \phi({\cal O}; \P)
 & = \E \Bigg\{ \frac{I(A=a)}{\bar\pi(A; L)} \frac{\Delta}{K(X-; A, L)} N(t) - \frac{I(A=a)}{\pi(A; L)} \frac{\Delta}{K(X-; A, L)} N(t) \Bigg\} \nonumber \\
 & = \Ef \Bigg\{ \frac{I(A=a)}{\bar\pi(a; L)} \frac{I(C^*_a \geq T^*_a)}{\bar{K}(T^*_a-; a, L)} N^*_a(t) - \frac{I(A=a)}{\pi(a; L)} \frac{I(C^*_a \geq T^*_a)}{K(T^*_a-; a, L)} N^*_a(t) \Bigg\} \nonumber \\
 & = \Ef \Bigg\{ \frac{I(A=a)}{\bar\pi(a; L)} \frac{K(T^*_a-; a, L)}{\bar{K}(T^*_a-; a, L)} N^*_a(t) - \frac{I(A=a)}{\pi(a; L)} N^*_a(t) \Bigg\} \nonumber \\
 & = \Ef \left[ N^*_a(t) \frac{I(A=a)}{\pi(a; L)} \left\{ \frac{K(T^*_a-; a, L)}{\bar{K}(T^*_a-; a, L)} \frac{\pi(a; L)}{\bar\pi(a; L)} - 1 \right\} \right]. \label{eq:supp:rem1}
\end{align}
Note, we delay marginalizing some of the terms to aid the calculation later. 

The last summand simplifies as
\begin{align}
 & \E D_{\text{aug}}({\cal O}; \bP) \nonumber \\
 =\, & \E \left[ - \frac{I(A=a) - \bar\pi(a; L)}{\bar\pi(a; L)} \bar{F}(t; a, L) 
   + \frac{I(A=a)}{\bar\pi(a; L)} \int_{(0,\infty)} \frac{\d N_C(u) - Y^{\dagger}(u) \,\d \bar\Lambda_{C}(u; a, L)}{\bar{K}(u; a, L)} \frac{\bar{F}(u, t; a, L)}{\bar{H}(u; a, L)} \right] \nonumber \\
 =\, & \E \Bigg[ - \frac{I(A=a) - \bar\pi(a; L)}{\bar\pi(a; L)} \bar{F}(t; a, L) \nonumber \\
   & \hspace{10mm} + \frac{I(A=a)}{\bar\pi(a; L)} \int_{(0,\infty)} \frac{\d \E\{N_C(u)\mid A=a, L\} - \E\{ Y^{\dagger}(u) \mid A=a, L\} \d \bar\Lambda_{C}(u; a, L)}{\bar{K}(u; a, L)} \frac{\bar{F}(u, t; a, L)}{\bar{H}(u; a, L)} \Bigg] \nonumber \\
 =\, & \E \biggl[ - \frac{I(A=a) - \bar\pi(a; L)}{\bar\pi(a; L)} \bar{F}(t; a, L) \nonumber
 \\ & \hspace{10mm} +
    \frac{I(A=a)}{\bar\pi(a; L)} \int_{(0,\infty)} \E\{ Y^{\dagger}(u) \mid A=a, L\} \frac{\d \Lambda_{C}(u; a, L) - d \bar\Lambda_{C}(u; a, L)}{\bar{K}(u; a, L)} \frac{\bar{F}(u, t; a, L)}{\bar{H}(u; a, L)} \biggr] \nonumber \\
 =\, & \E \Bigg[ - \frac{\pi(a; L) - \bar\pi(a; L)}{\bar\pi(a; L)} \bar{F}(t; a, L) \label{eq:supp:rem2} \\
   & \hspace{10mm} + \frac{\pi(a; L)}{\bar\pi(a; L)} \int_{(0,\infty)} B_1(u;a,L,\P,\bP)
  \bar{F}(u, t; a, L) \{\d \Lambda_{C}(u; a, L) - \d \bar\Lambda_{C}(u; a, L)\} \Bigg]. \label{eq:supp:rem3}
\end{align}
where
\[
B_1(u;a,L,\P,\bP) = \frac{H(u; a, L)}{\bar{H}(u; a, L)} \frac{K(u-; a, L)}{\bar{K}(u; a, L)}.
\]
Combining \cref{eq:supp:rem1} and \cref{eq:supp:rem2} and then rearranging, we find that 
\begin{align*}
    & \Ef \left[ N^*_a(t) \frac{I(A=a)}{\pi(a; L)} \left\{ \frac{K(T^*_a-; a, L)}{\bar{K}(T^*_a-; a, L)} \frac{\pi(a; L)}{\bar\pi(a; L)} - 1 \right\} \right]
     - \E \Bigg[ \frac{\pi(a; L) - \bar\pi(a; L)}{\bar\pi(a; L)} \bar{F}(t; a, L) \Bigg] \\
    =\, & \Ef \left[ \left\{\bar{F}(t; a, L) - \frac{I(A=a)}{\pi(a; L)} N^*_a(t) \right\} \frac{\bar\pi(a; L) - \pi(a; L)}{\bar\pi(a; L)}
     + \frac{I(A=a)}{\bar\pi(a; L)} N^*_a(t) \left\{ \frac{K(T^*_a-; a, L)}{\bar{K}(T^*_a-; a, L)} - 1 \right\} \right].
\end{align*}
The summand on the left equals
\begin{align}
     \Ef \biggl[  \biggl\{\bar{F}(t; a, L) - & \frac{I(A=a)}{\pi(a; L)} N^*_a(t) \biggr\} \frac{\bar\pi(a; L) - \pi(a; L)}{\bar\pi(a; L)}  \biggr] \nonumber \\
    & =\, \Ef \left[ \left\{\bar{F}(t; a, L) - F^*(t; a, L) \right\} \frac{\bar\pi(a; L) - \pi(a; L)}{\bar\pi(a; L)} \right]  \label{eq:supp:rem4};
\end{align}
by applying the Duhamel identity, the summand on the right equals
\begin{align}
    & \Ef \left[ \frac{I(A=a)}{\bar\pi(a; L)} N^*_a(t) \left\{ \frac{K(T^*_a-; a, L)}{\bar{K}(T^*_a-; a, L)} - 1 \right\} \right] \nonumber \\
    =\, & - \Ef \left[ \frac{I(A=a)}{\bar\pi(a; L)} N^*_a(t) \int_{(0,\infty)} I(T^*_a > u) \frac{K(u-; a, L)}{\bar{K}(u; a, L)} \{ \d  \Lambda_C(u; a, L) - \d  \bar\Lambda_C(u; a, L) \} \right] \nonumber \\
    =\, & - \E \left[ \frac{I(A=a)}{\bar\pi(a; L)} \int_{(0,\infty)} F(u, t; a, L) \frac{K(u-; a, L)}{\bar{K}(u; a, L)} \{ \d  \Lambda_C(u; a, L) - \d  \bar\Lambda_C(u; a, L) \} \right] \nonumber \\
    =\, & - \E \left[ \frac{\pi(a; L)}{\bar\pi(a; L)} \int_{(0,\infty)} F(u, t; a, L) \frac{K(u-; a, L)}{\bar{K}(u; a, L)} \{ \d  \Lambda_C(u; a, L) - \d  \bar\Lambda_C(u; a, L) \} \right]  \label{eq:supp:rem5},
\end{align}
where we recall that $F$ identifies $F^*(u, t; a, L) = \E^* \{ I(T^*_a > u) N^*_a(t) \mid A=a, L \}$. 

Define
\[
B_2(u,t;a,L,\P,\bP) =
 \frac{\bar{F}(u, t; a, L)}{\bar{H}(u; a, L)}  - \frac{F(u, t; a, L)}{H(u; a, L)};
\]
then, combining \cref{eq:supp:rem3} and \cref{eq:supp:rem5} and then rearranging, we find that 
\begin{align*}
    & \E \Bigg[ \frac{\pi(a; L)}{\bar\pi(a; L)} \int_{(0,\infty)} \frac{H(u; a, L)}{\bar{H}(u; a, L)} \frac{K(u-; a, L)}{\bar{K}(u; a, L)} \bar{F}(u, t; a, L) \{\d \Lambda_{C}(u; a, L) - \d  \bar\Lambda_{C}(u; a, L)\} \Bigg] \\
     & \hspace{15mm} - \E \left[ \frac{\pi(a; L)}{\bar\pi(a; L)} \int_{(0,\infty)} F(u, t; a, L) \frac{K(u-; a, L)}{\bar{K}(u; a, L)} \{ \d  \Lambda_C(u; a, L) - \d  \bar\Lambda_C(u; a, L) \} \right] \\
    =\, & \E \biggl[ \frac{\pi(a; L)}{\bar\pi(a; L)} \int_{(0,\infty)} H(u; a, L) 
    B_2(u,t;a,L,\P,\bP)
     \frac{K(u-; a, L)}{\bar{K}(u; a, L)} \{ \d  \Lambda_{C}(u; a, L) - \d  \bar\Lambda_{C}(u; a, L)\} \biggr] \\
    =\, & \E \left[ \frac{\pi(a; L)}{\bar\pi(a; L)} \int_{(0,\infty)} H(u; a, L) B_2(u,t;a,L,\P,\bP)
     \d \left\{ \frac{K(u; a, L) - \bar{K}(u; a, L)}{\bar{K}(u; a, L)} \right\} \right] \\
    =\, & \E \left[ \frac{\pi(a; L)}{\bar\pi(a; L)} \int_{(0,\infty)} H(u; a, L) B_2(u,t;a,L,\P,\bP)
     \d \left\{ \frac{K(u; a, L)}{\bar{K}(u; a, L)} \right\} \right],
\end{align*}
where the second-to-last identity again follows from the Duhamel equation. 
The remainder for $\mu_a(t)$ follows by combining this expression with \cref{eq:supp:rem4}.

An analogous calculation shows that the remainder for $\eta_a(t)$ is
\begin{align*}
 & \E \Bigg[ \{\bar{H}(t; a, L) - H(t; a, L)\} \frac{\bar\pi(a; L) - \pi(a; L)}{\bar\pi(a; L)} \\
 & \hspace{15mm} + \frac{\pi(a; L)}{\bar\pi(a; L)} \int_{(0,\infty)} H(u; a, L) \left\{ \frac{\bar{H}(t \vee u; a, L)}{\bar{H}(u; a, L)} - \frac{H(t \vee u; a, L)}{H(u; a, L)} \right\} \d \left\{ \frac{K(u; a, L)}{\bar{K}(u; a, L)} \right\} \Bigg] \\
 =\, & \E \Bigg[ \{\bar{H}(t; a, L) - H(t; a, L)\} \frac{\bar\pi(a; L) - \pi(a; L)}{\bar\pi(a; L)} \\
 & \hspace{15mm} + \frac{\pi(a; L)}{\bar\pi(a; L)} \int_{(0,t)} H(u; a, L) \left\{ \frac{\bar{H}(t; a, L)}{\bar{H}(u; a, L)} - \frac{H(t; a, L)}{H(u; a, L)} \right\} \d \left\{ \frac{K(u; a, L)}{\bar{K}(u; a, L)} \right\} \Bigg].
\end{align*}

The above arguments establish that the remainder term for each of the von Mises expansions given in the statement of the
Theorem depend on pairwise products of  $\bar K - K,$ $\bar H-H,$ and/or  $\bar F - F;$ hence, each is of second order.
\end{proof}

\section{Supplementary material for \cref{sec:estimation}}

\subsection{Proofs: Class of Influence Functions}
\label{sec:supp:if-class}

We start by proving the lemma which reduces the problem of finding the class of influence functions when the model is defined with potential censoring variables $C^{\full}_0,C^{\full}_1$ to the (simpler) problem of finding the class of influence functions when the model is defined with the single potential censoring variable $C^{\full}$. 

\begin{proof}[Proof of \cref{lem:single-reduction}]
 The set of augmentation terms in $\M(\pi, K)$ spans $h(L,A,X,\Delta)$ where the function $h$ satisfies 
 \begin{equation*}
     \Ef [ h\{L,A,T^{\full}_A \wedge C^{\full}_A, I(T^{\full}_A \leq C^{\full}_A), N^*_A(\cdot \wedge C^*)\} \mid L, T^{\full}_0, T^{\full}_1, N^*_0(), N^*_1() ] = 0.
 \end{equation*}
 The expectation in this condition is over the joint distribution $(A, C^{\full}_A) \mid (L, T^{\full}_0, T^{\full}_1, N^*_0(), N^*_1())$ which is characterized by the functions
 \begin{align*}
  & \P^*(A=1 \mid L, T^{\full}_0, T^{\full}_1, N^*_0(), N^*_1()) = \pi(1; L), \\
  & \P^*(C^{\full}_A > u \mid A, L, T^{\full}_0, T^{\full}_1, N^*_0(), N^*_1()) = K(u; A, L),
 \end{align*}
where we recall that $K(u; A, L)$ is the identification of $\Pf(C^{\full}_A > u \mid A, L)$.

 The set of augmentation terms in $\M_{\text{single}}(\pi_0, K_0)$ spans $h_{\text{single}}(L,A,X,\Delta)$ where the function $h_{\text{single}}$ satisfies 
 \begin{equation*}
     \Ef [ h_{\text{single}}\{L,A,T^{\full}_A \wedge C^*, I(T^{\full}_A \leq C^*), N^*_A(\cdot \wedge C^*) \} \mid L, T^{\full}_0, T^{\full}_1, N^*_0(), N^*_1() ] = 0.
 \end{equation*}
 The expectation in this condition is over the joint distribution $(A, C^*) \mid (L, T^{\full}_0, T^{\full}_1)$ which is characterized by the functions
 \begin{align*}
  & \Pf(A=1 \mid L, T^{\full}_0, T^{\full}_1, N^*_0(), N^*_1()) = \pi(1; L), \\
  & \Pf(C^* > u \mid A, L, T^{\full}_0, T^{\full}_1, N^*_0(), N^*_1()) = K_{\text{single}}(u; A, L),
 \end{align*}
 where we recall that $K_{\text{single}}(u; A, L)$ is the identification of $\Pf(C^* > u \mid A, L)$.

 The result follows since these sets are the same. The equivalence of results in efficiency theory (such as the efficient influence function) follows since the likelihoods are the same.
\end{proof}

We now calculate the class of influence functions. We rely on the following result, which is given by \citet[Theorem 1.3]{van2003unified} and \citet[Theorem 8.3]{tsiatis2006semiparametric}. We state the theorem in its full generality for any coarsened data.

\begin{lem}
\label{lem:class-rep}
 Consider a coarsened data model. Assume that coarsening at random holds and that $\Pf(\text{coarsening variables} \mid \text{full data})$ is known. Define $\varphi^F$ as the full data influence function, assumed unique. Then the class of observed data influence functions is 
 \begin{equation*}
  U(\varphi^F) + \mathrm{AS},
 \end{equation*}
 where $\Ef\{U(\varphi^F) \mid \text{full data}\} = \varphi^F$ so that $U$ maps the full data influence function into the observed data and $\mathrm{AS} = \{\text{observed data functions } \, h \,:\, \Ef( h \mid \text{full data} )\} = 0$ is the augmentation space.
\end{lem}
The derivation of the class of influence functions simply characterizes the terms in the preceding lemma.

\begin{proof}[Proof of \cref{prop:class}]
 Throughout we consider the simpler model $\M(\pi,K_{\text{simpler}})$ in which the coarsening variables are $(A,C^*)$. After the class of influence functions is derived in this model, the class of influence functions in $\M(\pi, K)$ follows immediately from \cref{lem:single-reduction}.
 
 We start by applying \cref{lem:class-rep}. 
 The unique full-data influence function $\varphi^F$ is straightforward to find. The inverse probability weighted mapping 
\begin{equation*}
 \frac{I(A=a)}{\pi(a; L)} \frac{\Delta}{K_{\sing}(X; a, L)} \varphi^F
\end{equation*}
satisfies the conditions for $U$ in \cref{lem:class-rep}. The augmentation space is the set of functions $h\bigl(L,A,X,\Delta,N() \bigr)$ such that
\begin{equation*}
\Ef \Bigl[ h \Bigl\{L,A,T^{\full}_A \wedge C^*, I(T^{\full}_A \leq C^*), N^*_A(\cdot \wedge C^*) \Bigr\} \Big| L, T^{\full}_0, T^{\full}_1, N^*_0(), N^*_1() \Bigr] = 0
\end{equation*}
and is more difficult to characterize. 

Define the intermediate-data as 
\begin{equation*}
    \bigl( L, A, T^{\full} = T^{\full}_A, N^{\full}() = N^{\full}_A() \bigr).
\end{equation*}
Notice that the observed data 
may be written as $\Phi_{\text{obs}}\{\Phi_{\text{int}}(L,T^{\full}_0,T^{\full}_1,N^{\full}_0(),N^{\full}_1(); C^{\full}) ; A\}$, where
\begin{align*}
 \Phi_{\text{int}}(L,T^{\full}_0, T^{\full}_1, N^{\full}_0(), N^{\full}_1(); A) 
 & = \bigl(L, A, T^{\full}=T^{\full}_A, N^{\full}()=N^{\full}_A() \bigr), \\
 \Phi_{\text{obs}}(L,A,T^{\full},N^{\full}(); C^{\full}) 
 & = \bigl( L,A,X=T^{\full}\wedge C^{\full},\Delta=I(T^{\full}\leq C^{\full}), N()=N^*_A(\cdot\wedge C^*) \bigr).
\end{align*}
Denote the observed data as $\mO$, the full data as $Z^*$, and the intermediate data as $I^*_{\mathrm{data}}$.

The coarsening at random condition for $\Phi_{\text{int}}$ is that $(T^{\full}_0,T^{\full}_1,N^{\full}_0(),N^{\full}_1()) \ind A \mid L$, and the coarsening at random condition for $\Phi_{\text{obs}}$ is that $(T^{\full}_a,N^{\full}_a()) \ind C^{\full} \mid A=a, L$ on $T^{\full}_a < C^{\full}$. Both conditions are implied by the coarsening at random assumption for $\Phi$ which simultaneously coarsens. Similarly, the positivity assumption for $A$ and for $C^* \mid A$ is implied by the positivity assumptions we have made.

Define $\mathcal{G}$ as the set of intermediate-data functions $g$ such that $\Ef \{ g(I^*_{\mathrm{data}}) \mid Z^* \} = 0$, and define $\mathcal{F}$ as the set of observed-data functions $f$ such that $\Ef \{ f(\mO) \mid I^*_{\mathrm{data}} \} = 0$. \citet{van2003unified, tsiatis2006semiparametric} show that
\begin{equation*}
 \mathcal{G} 
 = \Bigl\{ \{I(A=a)-\pi(a; L)\}h_1(L) \, : \, h_1 \Bigr\}.
\end{equation*}
\cref{lem:aux-surv-class} below shows that 
\begin{equation*}
 \mathcal{F} 
 = \left\{ \int_{(0,\infty)} h_2\left( u; \bar{N}(u), A, L \right) \d M_{C,\textrm{single}}(u; A, L)  \, : \, h_2 \right\},
\end{equation*}
where $\bar{N}(u) = \{N(v) : 0<v\leq u \}$. Notice the index function $h_2$ depends on the history of the recurrent event process $N$ up to and including time $u$.

We find the augmentation term by exploiting the identity
\begin{equation*}
 \Ef \Bigl[ h(\mO) \Big| Z^* \Bigr]
 = \Ef \Bigl[ \Ef \Big\{ h(\mO) \Big| I^*_{\mathrm{data}} \Big\} \Big| Z^* \Bigr].
\end{equation*}
By Lemma 7.4 in \citet{tsiatis2006semiparametric}, the augmentation space is 
\begin{equation*}
 \frac{\Delta}{K_{\sing}(X-; A, L)} \mathcal{G} + \mathcal{F}.
\end{equation*}

We now plug $\mathcal{G}$ and $\mathcal{F}$ into this expression and rearrange the terms. By \cref{lem:rr-lems}, a typical element of the augmentation space is
\begin{align*}
    & \frac{\Delta}{K_{\sing}(X; L, A)} \big\{ I(A=a)-\pi(a; L) \big\} h_1(L) + \int_0^\infty \d M_{C,\sing}(u, A, L) h_2 (u; L, A, \bar{N}(u)) \\
    =\, & \left\{ 1 - \int_0^\infty \frac{\d M_{C,\sing}(u; A, L)}{K_{\sing}(u; L, A)} \right\} \big\{ I(A=a)-\pi(a; L) \big\} h_1(L) \\
    & \hspace*{10mm} + \int_0^\infty \d M_{C,\sing}(u; A, L) h_2 (u; L, A, \bar{N}(u)) \\
    =\, & \big\{ I(A=a)-\pi(a; L) \big\} h_1(L) \\
    & \hspace*{10mm}
    + \int_0^\infty \d M_{C,\sing}(u; A, L) \left\{ h_2 (u; L, A, \bar{N}(u)) - \frac{\{ I(A=a)-\pi(a; L) \big\} h_1(L)}{K_{\sing}(u; A, L)} \right\} \\
    =\, & \big\{ I(A=a)-\pi(a; L) \big\} h_1(L) + \int_0^\infty \d M_{C,\sing}(u; A, L) \tilde{h}_2(u; L, A, \bar{N}(u)),
\end{align*}
where $\tilde{h}_2(u; L, A, \bar{N}(u))$ is arbitrary.
\end{proof}

We now state a prove an auxilary lemma used in the proof of \cref{prop:class}.

\begin{lem}
\label{lem:aux-surv-class}
    Under right-censoring and coarsening at random, a typical augmentation term is $\int_{(0,\infty)} h_2\left( u; \bar{N}(u), A, L \right) \d M_{C,\textrm{single}}(u; A, L)$.
\end{lem}

\begin{proof}
    The argument closely follows that in \citet[Chapter 9.3]{tsiatis2006semiparametric}, except it treats the discrete case exactly. It also related to \citet[Theorem 1.1]{van2003unified}.

    Define the coarsening variable $\mC \in(0,\infty)$ as 
    \begin{equation*}
        \mC = X I(\Delta = 0) + \infty I(\Delta=1).
    \end{equation*}
    This variable explains the coarsening pattern and its structure reveals that the coarsening in $\Phi_{\mathrm{obs}}$ is monotone. The observed data is $r$ and
    \begin{equation*}
        G_r
        := \begin{cases}
            \bigl( L, A, 
             T^*>r, \bar{N}^*(r) \bigr) \text{ if } r < \infty \\
            \bigl( L, A, 
            T, \bar{N}^*() \bigr) \text{ if } r = \infty.
        \end{cases} 
    \end{equation*}

    For $r<\infty$, the coarsening hazard function is defined as
    \begin{equation*}
        \lambda_r^*
        := \P \bigl(\mC = r \mid \mC \geq r, L, A, T^*>r, \bar{N}^*(r) \bigr),
    \end{equation*}
    where the expression in the conditioning set is simply $G_r$. Considering the definition of $\mC$ shows that this expression simplifies as
    \begin{align*}
        \lambda_r^* 
         = I(T^*>r) \P \bigl(C^* = r & \mid C^* \geq r, L, A, T^*>r, \bar{N}^*(r) \bigr) 
        \\ & 
        = I(T^*>r) \P \bigl(C^* = r \mid C^* \geq r, L, A, T^*>r \bigr).
    \end{align*}
    Using Lemma 1 in \citet{baer2025survival}, the coarsening hazard $\lambda_r^*$ is observed to be $I(T^*>r) \lambda_C(r)$, where $\lambda_C$ is the identified counting-hazard for censoring. 

    Next, \citet[Theorem 9.2]{tsiatis2006semiparametric} presents that a typical element of the augmentation space is
    \begin{align*}
        \sum_{r < \infty} \frac{I(\mC = r) - I(\mC \geq r) \lambda_r^*}{K_r} h(r; G_r) 
        =\, & \sum_{r < \infty} \frac{dN_C(r) - I(\mC \geq r, T^*>r) \lambda_C(r)}{K(r; A, L)} h(r; G_r) \\
        =\, & \sum_{r < \infty} \frac{dN_C(r) - Y^{\dagger}(r) \lambda_C(r)}{K(r; A, L)} h(r; G_r).
    \end{align*}
    Although $G_r$ nominally depends on $T^*$, the dependence is removed due to the $dN_C$ and $Y^{\dagger}$. 

    The result follows in the discrete case. The general case follows from \citet[Lemma 5.2]{van2004robins}.
\end{proof}

\begin{proof}[Proof of \cref{thm:est-master}]
 The proof will readily follow from the von Mises decompositions in \cref{prop:vonmises} and the techniques described by \citet{kennedy2023semi}. Throughout we study the asymptotics of the estimator $\hat\psi_n;$ the behavior of $\hat\psi_{n,\mathrm{proj}}$ follows from results in \citet{westling2020correcting}, as recently used by \citet{westling2021inference} in a failure time setting.

Recall that the estimand $\psi_0$ has two component blocks; we establish asymptotic linearity of $\hat\psi_{n}$ by studying each component separately. 
For these components at a given time $t>0$, the estimator may be written as $\E_n \phi_n$ for some data-dependent function $\phi_n$ and satisfies the decomposition:
 \begin{align*}
     \E_n \phi_n - \E_0 \phi_0
     = \E_n (I - \E_0) \phi_0 + (\E_n - \E_0) (\phi_n - \phi_0) + \E_0 (\phi_n - \phi_0),
 \end{align*}
 where $\E_0 \phi_0$ is the estimand. The first term is the desired influence function $D$, and the second term is $o_p(n^{-1/2})$ by \citet[Lemma 2]{kennedy2020efficient}. Thus we focus on the third term. We can straightforwardly show that $\E_0 (\phi_n - \phi_0)$ equals the remainder
 derived in the proof of \cref{prop:vonmises}, where $\P = \P_0$ and $\bar{\P} = \P_n$.
In particular, in the case of $\mu,$ we obtain
the remainder term
\begin{align*}
 r_{n,\mu}(t) & =  \E_0 \Bigg[ \Bigl\{F_n(0, t; a, L) - F_0(0, t; a, L) \Bigr\} \frac{\pi_n(a; L) - \pi_0(a; L)}{\pi_n(a; L)} \\
   + &  \frac{\pi_0(a; L)}{\pi_n(a; L)} \int_{(0,\infty)} H_0(u; a, L) \left\{ \frac{F_n(u, t; a, L)}{H_n(u; a, L)} - \frac{F_0(u, t; a, L)}{H_0(u; a, L)} \right\} \d \left\{ \frac{K_0(u; a, L) - K_n(u; a, L)}{K_n(u; a, L)} \right\} \Bigg].
\end{align*}
The integral term in the second summand, while nominally
over $(0,\infty),$ is always on a finite interval 
$(0,\tau_{0}(a,L)],$ where $\tau_{0}(a,L) \leq \tau$ almost surely due to \cref{asp:trunc}. 
For later use, we also define the total variation function
\[
\mathrm{TV}_{K_0/K_n}(u;a,L) =
\sup \left\{
\sum_{k=1}^\ell \left| \frac{K_0(s_k; a, L)}{K_n(s_k; a, L)}
- \frac{K_0(s_{k-1}; a, L)}{K_n(s_{k-1}; a, L)}
\right| : 0 < s_1 < \cdots < s_{\ell} = u
\right\}
\]
where the supremum is taken over all
partitions of $[0,u].$ 

Below, it is  established that the relevant remainder term
in the case of $m \geq 1$ landmark times satisfies
$| r_{\mu, n} | = O(\sum_{j=1}^m \{ r_1(t_j) + r_2(t_j) + r_3 \}),$ where
    \begin{align*}
        r_1(t) & = \Big\| F_n(0, t; a, L) - F_0(0, t; a, L) \Big\|_{L_2(\P_0)} \Big\| \pi_n(a; L) - \pi_0(a; L) \Big\|_{L_2(\P_0)}, \\
        r_2(t) & = \bigg\| \sup_{u \in (0,\infty)} \left\{ \frac{F_n(u, t; a, L) - F_0(u, t; a, L)}{F_0(u, t; a, L)} \right\} \bigg\|_{L_2(\P_0)} \bigg\| 
        \mathrm{TV}_{K_0/K_n}(\tau;a,L)
        \bigg\|_{L_2(\P_0)},\\
        r_3 & = \bigg\| \sup_{u \in (0,\infty)} \left\{ \frac{H_n(u; a, L) - H_0(u; a, L)}{H_0(u; a, L)} \right\} 
        \bigg\|_{L_2(\P_0)} 
        \bigg\| 
        \mathrm{TV}_{K_0/K_n}(\tau;a,L)
        \bigg\|_{L_2(\P_0)}. 
    \end{align*}
If each of these terms is $o_p(n^{-1/2}),$ the stated result follows for estimating the component corresponding to $\mu$. It can be seen that achieving this rate depends directly on the rate at which each nuisance parameter can be estimated. 

We approach each summand in $r_{n,\mu}(t)$ separately.
For the first summand, and under Assumption~2 in \cref{thm:est-master}, we know that
\begin{align*}
     \bigg| \E_0 \Bigg[ & \Bigl\{F_n(0, t; a, L) - F_0(0, t; a, L) \Bigr\} \frac{\pi_n(a; L) - \pi_0(a; L)}{\pi_n(a; L)} \Bigg] \bigg| \\
    \leq \, & \frac{1}{\epsilon'} \E_0 \bigg[ \Big| \{F_n(0, t; a, L) - F_0(0, t; a, L)\Bigr\} \Bigl\{\pi_n(a; L) - \pi_0(a; L)\Bigr\} \Big| \biggl] \\
    \leq \, & \frac{1}{\epsilon'} \Big\| F_n(0, t; a, L) - F_0(0, t; a, L) \Big\|_{L_2(\P_0)} \Big\| \pi_n(a; L) - \pi_0(a; L) \Big\|_{L_2(\P_0)}.
\end{align*}
For the second summand, we have
\begin{align}
    & \left| \E_0 \Bigg[ \frac{\pi_0(a; L)}{\pi_n(a; L)} \int_{(0,\infty)} H_0(u; a, L) \left\{ \frac{F_n(u, t; a, L)}{H_n(u; a, L)} - \frac{F_0(u, t; a, L)}{H_0(u; a, L)} \right\} \d \left\{ \frac{K_0(u; a, L)}{K_n(u; a, L)} \right\} \Bigg] \right| \nonumber \\
    \leq \, & \frac{1}{\epsilon'} \E_0 \Bigg[ \left| \int_{(0,\infty)} H_0(u; a, L) \left\{ \frac{F_n(u, t; a, L)}{H_n(u; a, L)} - \frac{F_0(u, t; a, L)}{H_0(u; a, L)} \right\} \d \left\{ \frac{K_0(u; a, L)}{K_n(u; a, L)} \right\} \right| \Bigg] \nonumber \\
    \leq \, & \frac{1}{\epsilon'} \E_0 \Bigg[ \int_{(0,\infty)} \left|  \left\{ \frac{F_n(u, t; a, L)}{H_n(u; a, L)} - \frac{F_0(u, t; a, L)}{H_0(u; a, L)} \right\} \right| \d \mathrm{TV}_{K_0/K_n}(u; a,L)
\Bigg] \nonumber \\
    \leq \, & \frac{1}{\epsilon'} \E_0 \Bigg[ \sup_{u \in (0,\infty)} \left| \frac{F_n(u, t; a, L)}{H_n(u; a, L)} - \frac{F_0(u, t; a, L)}{H_0(u; a, L)} \right| 
    \mathrm{TV}_{K_0/K_n}(\tau;a,L)
\Bigg] \nonumber \\
    \leq \, & \frac{1}{\epsilon'} \left\| \sup_{u \in (0,\infty)} \left| \frac{F_n(u, t; a, L)}{H_n(u; a, L)} - \frac{F_0(u, t; a, L)}{H_0(u; a, L)} \right| \right\|_{L_2(\P_0)} 
    \left\| 
    \mathrm{TV}_{K_0/K_n}(\tau;a,L)
    \right\|_{L_2(\P_0)}. \label{eqn:rem1}
\end{align}

Define $\iota \in \{1,2\}$ as $\iota = \arg\max\{ \sup_{H_n(u; a, L)>0} u, \sup_{H_0(u; a, L)>0} u \}$. 
Denote $H_{\mathrm{max}}(u; a, L) = I(\iota=1) H_n(u; a, L) + I(\iota = 2) H_0(u; a, L)$ and similarly $F_{\mathrm{max}}(u, t; a, L) = I(\iota=1) F_n(u; a, L) + I(\iota = 2) F_0(u; a, L)$. Denote $H_{-\mathrm{max}}, F_{-\mathrm{max}}$ as the other functionals. This notation may seem complicated: we introduce it only to ensure that we do not divide by zero in the following displayed equation. In particular, the first term in the last expression in \cref{eqn:rem1} may further be decomposed as
\begin{align*}
    & \Bigg\| \sup_{u \in (0,\infty)} \left| \frac{F_n(u, t; a, L)}{H_n(u; a, L)} - \frac{F_0(u, t; a, L)}{H_0(u; a, L)} \right| \Bigg\|_{L_2(\P_0)} \\
    \leq \, & \Bigg\| \sup_{u \in (0,\infty)} \left| \frac{F_n(u, t; a, L) - F_0(u, t; a, L)}{H_{\mathrm{max}}(u; a, L)} - \frac{F_{-\mathrm{max}}(u, t; a, L)
    \Big\{ H_n(u; a, L) - H_0(u; a, L) \Big\}}{H_{\mathrm{max}}(u; a, L) H_{-\mathrm{max}}(u; a, L)} 
     \right| \Bigg\|_{L_2(\P_0)} \\
    \leq \, & (I) + (II),
\end{align*}
where $(I)$ and $(II)$ are defined after applying the triangle inequality in the natural way; see also below for explicit definitions.  Similarly to before: although the supremum is nominally over $(0,\infty),$ it is in reality formally over the interval whose upper limit is the smallest $u$ such that $H_{\mathrm{max}}(u; a, L) = 0$. 

The summands may be upper bounded as
\begin{align*}
    (I)
    :=\, & \Bigg\| \sup_{u \in (0,\infty)} \left| \frac{F_n(u, t; a, L) - F_0(u, t; a, L)}{H_{\mathrm{max}}(u; a, L)K_n(u; a, L)} \right| \Bigg\|_{L_2(\P_0)} \\
    \leq \, & \frac{1}{\epsilon'} \Bigg\| \sup_{u \in (0,\infty)} \left| \frac{F_n(u, t; a, L) - F_0(u, t; a, L)}{H_{\mathrm{max}}(u; a, L)} \right| \Bigg\|_{L_2(\P_0)} \\
    \leq \, & \frac{1}{\epsilon'} \Bigg\| \sup_{u \in (0,\infty)} \left| \frac{F_{\mathrm{max}}(u; a L)}{H_{\mathrm{max}}(u; a, L)} \frac{F_n(u, t; a, L) - F_0(u, t; a, L)}{F_{\mathrm{max}}(u; a L)} \right| \Bigg\|_{L_2(\P_0)} \\
    \leq \, & \frac{C_F \vee C_F'}{\epsilon'} \Bigg\| \sup_{u \in (0,\infty)} \left| \frac{F_n(u, t; a, L) - F_0(u, t; a, L)}{F_{\mathrm{max}}(u; a L)} \right| \Bigg\|_{L_2(\P_0)}, 
\end{align*}
and
\begin{align*}
    (II)
    :=\, & \Bigg\| \sup_{u \in (0,\infty)} \left| \frac{F_{-\mathrm{max}}(u, t; a, L)}{H_{\mathrm{max}}(u; a, L) H_{-\mathrm{max}}(u; a, L) K_n(u; a, L)} \Big\{ H_n(u; a, L) - H_0(u; a, L) \Big\} \right| \Bigg\|_{L_2(\P_0)} \\
    \leq \, & \Bigg\| \sup_{u \in (0,\infty)} \left| \frac{F_{-\mathrm{max}}(u, t; a, L)}{H_{-\mathrm{max}}(u; a, L) K_n(u; a, L)} \right| \sup_{u \in (0,\infty)} \left| \frac{H_n(u; a, L) - H_0(u; a, L)}{H_{\mathrm{max}}(u; a, L)} \right| \Bigg\|_{L_2(\P_0)} \\
    \leq \, & \frac{C_F \vee C_F'}{\epsilon'} \Bigg\| \sup_{u \in (0,\infty)} \left| \frac{H_n(u; a, L) - H_0(u; a, L)}{H_{\mathrm{max}}(u; a, L)} \right| \Bigg\|_{L_2(\P_0)}. 
\end{align*}
For simplicity, we consider $\iota=2$ in the statement of the theorem. 
It can now be seen that these upper bounds are respectively exactly those expressions appearing in the 
upper bounds for the estimation remainder $| r_{\mu, n} |;$ upon applying these bounds for each landmark time, the triangle
inequality completes the proof in the case of $\mu$. 

The result for the counterfactual failure survival function $\eta$ follows similarly; starting with the analogous
remainder term derived from that in the proof
of \cref{prop:vonmises}, it can be shown that
$| r_{\eta, n} | = O(\sum_{j=1}^m \{ r_4(t_j) + r_5(t_j) + r_3 \}),$ where
    \begin{align*}
        r_4(t) & = \Big\| H_n(t; a, L) - H_0(t; a, L) \Big\|_{L_2(\P_0)} \Big\| \pi_n(a; L) - \pi_0(a; L) \Big\|_{L_2(\P_0)}, \\
        r_5(t) & = \bigg\| \sup_{u \in (0,\infty)} \left\{ \frac{H_n(t \vee u; a, L) - H_0(t \vee u; a, L)}{H_0(t \vee u; a, L)} \right\} \bigg\|_{L_2(\P_0)} 
        \left\| 
    \mathrm{TV}_{K_0/K_n}(\tau;a,L)
    \right\|_{L_2(\P_0)}.
    \end{align*}
\end{proof}

\section{Supplementary material for \cref{sec:litreview}}
\label{sec:supp:litreview:steele}

In this section, we discuss the estimator proposed in \citet{su2020doubly}; in particular, we consider the augmented oracle estimator for the $a-$specific mean function given by
\begin{align}
\label{eq:su-est}
   & \E_n \int_{(0,t]} \Bigg\{ \frac{I(A=a)}{\pi_0(a; L)} \frac{\d N(u)}{\E\{K_0(u-; a, L)\}} -
    \left\{ \frac{I(A=a) I(C \geq u)}
    {\pi_0(a; L) \E_0\{K_0(u-; a, L)\}} 
    -1\right\} \Bigg\} 
    \,\d F_0(0, u; a, L)
.
\end{align}
It is easy to see that this estimator is in the same form as that considered in \citet{su2020doubly}; as
expressed above, it is somewhat more general since it does not rely on the specification of 
(semi-)parametric working models for $\pi_0(a;L)$ and $dF_0(0,u;a,L).$ 

As noted in the main paper, \citet{su2020doubly} state that their  augmented estimator is doubly robust under the following assumptions on censoring: $(C^*_0, C^*_1, N^*_0(\cdot), N^*_1(\cdot)) \ind A \mid L$ and $C^*_a \ind N^*_a(\cdot) \mid L$ for $a=0,1$. However, 
as will be shown in \cref{sec:su-review}, their claims of consistency and doubly robustness also rely on the unstated assumption that $C^*_a \ind L \mid A=a$ for $a=0,1$. In the case where this relatively strong assumption is violated, it follows that their inverse probability weighted estimator (i.e., the first term in \eqref{eq:su-est}) is inconsistent and that their augmented estimator is generally not doubly robust. In \cref{sec:su-model}, we show that a simple modification of their estimator repairs the inconsistency, and and demonstrate
that the corresponding influence function falls within the class 
given in \cref{prop:class} (i.e., properly interpreted in the absence
of failure). The resulting estimator is also doubly robust in the sense that consistent estimation of $F_0,$ or of $(\pi_0,K_0),$
results in a consistent estimator of $\psi_0.$ 

\subsection{Review}
\label{sec:su-review}
We start by considering the claims that their inverse probability weighted estimator is consistent and that their augmented estimator is doubly robust. On Page 2 of their supplement, and expressed
in terms of our notation,
they claim that
\begin{align*}
    \E\bigl[ \E^* \left\{ d N^*_a(u) \mid A=a, L \right\} & \P^* \left( C^*_a \geq u \mid A=a, L \right) \bigr] \\
    &
    = \E\left[ \P^* \left( C^*_a \geq u \mid A=a, L \right) \right] \times \E\left[ \E^* \left\{ d N^*_a(u) \mid A=a, L \right\} \right];
\end{align*}
however, unless one or both of $\E^* \left\{ d N^*_a(u) \mid A=a, L \right\}$ or $\P^* \left( C^*_a \geq u \mid A=a, L \right)$ is not a function of $L,$ this cannot hold in general.
This is the first place where the implicit assumption $C^*_a \ind L \mid A=a$ for each $a=0,1$  is used.

Let $d\tilde{F}(\cdot; a,\ell)$ be the probability limit of their outcome regression estimator when $A = a$ and $L = \ell$; in addition, define
$\tilde K_a(u-) := \E \{ K(u-; a, L) \}.$ On Pages 13-14 of their supplement, the following
expression is given for the (differential of) their augmentation term when $a=1:$
\begin{align*}
\E \left[  
\left\{
\frac{\frac{I(A=a)}{\pi(a; L)} I (C^*_a \geq u) - \tilde K_a(u-)}
{\tilde K_a(u-)}
\right\}
 \left\{ d N^*_a(u) - d\tilde{F}(u;a,L) \right\}
 \right].
\end{align*}
Using iterated expectation and the assumptions
explicitly made in \citet{su2020doubly},
this expression can be rewritten
as
\begin{align*}
\E \biggl[  & 
 \left\{ 
 \E \bigl( d N^*_a(u) \mid
 A = a, L \bigr) - d\tilde{F}(u;a,L) \right\}
\left\{
\frac{\frac{I(A=a)}{\pi(a; L)} I(C^*_a \geq u) - \tilde K_a(u-)}
{\tilde K_a(u-)}
\right\}
 \biggr] \\
& = \E \biggl[
 \left\{ 
 \E \bigl( d N^*_a(u) \mid
 A = a, L \bigr) - d\tilde{F}(u;a,L) \right\}
\left\{ \frac{K(u-; a,L) - \tilde K_a(u-)}{\tilde K_a(u-)}
\right\}
 \biggr] 
 \end{align*}
On Page 14 of their supplement, they claim this expectation is equal to zero. Although this readily follows when $K(u-; a, L)$ does not depend on $L$
(i.e., $K(u-; a, L) = \tilde K_a(u-),)$ this expectation is not zero in general because $K(u-; a, L) - \tilde K_a(u-)$ is not uncorrelated with every function of $L$.

The implicit assumption that $C^*_a \ind L \mid A=a$ for each $a=0,1$ has a testable parametric implication. Define the conditional and unconditional cumulative hazard for censoring as
\begin{align*}
    \Lambda_a^*(t; l) := \int_{(0,t]} \frac{\d \PP^*(C^*_a \leq u \mid A=a, L=l)}{\PP^*(C^*_a \geq u \mid A=a, L=l)}, \hspace{5mm}
    \Lambda_a^*(t) := \int_{(0,t]} \frac{\d \PP^*(C^*_a \leq u \mid A=a)}{\PP^*(C^*_a \geq u \mid A=a)}.
\end{align*}
Define $C := C^*_A;$ in \citet{su2020doubly},
$C$ is always observed. Now define the observed data functionals
\begin{align*}
    \Lambda_a(t; l) := \int_{(0,t]} \frac{\d \PP(C \leq u \mid A=a, L=l)}{\PP(C \geq u \mid A=a, L=l)}, \hspace{5mm}
    \Lambda_a(t) := \int_{(0,t]} \frac{\d \PP(C \leq u \mid A=a)}{\PP(C \geq u \mid A=a)}.
\end{align*}

It's clear from consistency that $\Lambda_a^*(t; l)=\Lambda_a(t; l)$ and $\Lambda_a^*(t)=\Lambda_a(t)$. Their second assumption immediately shows that $\Lambda_a^*(\cdot; l) = \Lambda_a^*(\cdot)$. Therefore, when $C^*_a \ind L \mid A=a,$ $\Lambda_a(\cdot; l) = \Lambda_a(\cdot)$. This  parametric constraint shrinks the observed data tangent space from nonparametric to semiparametric. Therefore the class of influence functions in their model is larger than the class of influence functions in our model \citep[Theorem 4.3]{tsiatis2006semiparametric}.

\subsection{A modified estimator and its influence function}
\label{sec:su-model}

The previous results demonstrate that the inverse probability weighted estimator is generally inconsistent and the augmented estimator is not doubly robust unless $C^*_a \ind L \mid A=a$ for each $a=0,1.$ The theory of this paper suggests instead that $\E \{ K(u-; a, L) \}$ in \eqref{eq:su-est}
should be replaced by  $K_0(u-; a, L)$; note that $\E \{ K(u-; a, L) \} = K_0(u-; a, L)$ when $C^*_a \ind L \mid A=a$ for each $a=0,1$ 
because neither then depends on $L.$ In the following proposition, the influence function  \eqref{eq:su-est}, modified so that
$\E \{ K(u-; a, L) \}$ is everywhere replaced by $K_0(u-; a, L),$ is shown to belong to the class of influence functions given in \cref{prop:class}.
We note that \citet{su2020doubly} do not explicitly impose distributional assumptions on $C^*_a;$ the calculations below allow $C^*_a$ to have a general distribution.

\begin{prop}
\label{prop:mod-su-IF0}
    The modified oracle estimator derived
    from \eqref{eq:su-est} has influence function 
\begin{align}
\nonumber
N^*_a(t) - \psi_0 + &
\frac{I(A=a)-\pi_0(a; L)}{\pi_0(a; L)} 
\left( N^*_a(t) - F_0(0,t;a,L) \right) \\
\nonumber
& + 
\frac{I(A=a)}{\pi_0(a; L)} 
\Biggl[
\int_{(0,t]} \left\{ 
N^*_a(u) - F_0(0,u;a,L) \right\} 
\frac{\d \tilde M_{C,0}(u;a,L)}{K_0(u; a, L)} \\
& 
\label{eq:su-IF-causal}
- \left\{ N^*_a(t) - F_0(0,t;a,L) \right\} 
\int_{(0,t]} \frac{\d \tilde M_{C,0}(u;a,L)}{K_0(u; a, L)} \Biggr],
\end{align}
where $
\d \tilde M_{C,0}(u;a,L) = \d N_C(u) - I(C \geq u) \d \Lambda_{C,0}(u; a, L).
$ It is doubly robust in the sense that consistent estimation of 
$F_0,$ or of $(\pi_0,K_0),$ leads to consistent estimation of
$\psi_0.$ 
\end{prop}  

The final form given for the influence function, or \eqref{eq:su-IF-causal}, is nonstandard but convenient
for evaluating double robustness; see, for example, \citet{bai2013doubly}, who uses a similar approach.
The next result establishes that \eqref{eq:su-IF-causal} is also an element of the class of influence functions
given in \cref{prop:class} (i.e., in the absence of failure). The proof proceeds by first showing how
a particular member of the class given in \cref{prop:class}, which unlike \citet{su2020doubly} allows for the additional
possibility of failure, can be put into a form analogous  to \eqref{eq:su-IF-causal}.

\begin{prop}
\label{prop:mod-su-IF1}
In the absence of failure, the influence function derived from the class given in \cref{prop:class} using the indices 
        $h_1(L)  = F_0(0, t; a, L) / \pi_0(a; L)$ and
        \[
        h_2(u; A, L, \bar{N}(u))  = \frac{I(A=a)}{\pi_0(a; L)} \frac{N(u\wedge t) - I(u \leq t) \left\{ F_0(0, u; a, L) - F_0(0, t; a, L) \right\}}{K_0(u; a, L)},
        \]
coincides with \eqref{eq:su-IF-causal}. 
\end{prop}
The proof of each result is given below. 
\begin{proof}[Proof of \cref{prop:mod-su-IF0}]
After replacing all nuisance parameter estimators with their true value and modifying $K$ as indicated previously, our modification of the augmented estimator in \citet{su2020doubly} may be written as
\begin{align*}
    \E_n \left[ \int_{(0,t]} \phi_u(\mO) - \Big\{ \mathbb{E} (\phi_u(\mO) \,\mid\, L, A, C) - \mathbb{E} (\phi_u(\mO) \,\mid\, L) \Big\} \right],
\end{align*}
where $\phi_u(\mO) = \frac{I(A=a)}{\pi_0(a; L)} \frac{I(C \geq u)}{K_0(u-; a, L)} \d N (u)$ is an IPW estimator for the recurrent event increment:
\begin{align*}
    \mathbb{E} \{\phi_u(\mO) \,\mid\, L, A, C\}
    & =  \frac{I(A=a)}{\pi_0(a; L)} \frac{I(C \geq u)}{K_0(u-; a, L)} \d \E^\full \{N^*_a (u) \mid A=a, L\}, \\
    \mathbb{E} \{\phi_u(\mO) \,\mid\, L\} & = \d \E^\full \{ N^*_a (u) \mid A=a, L\}.
\end{align*}

The uncentered influence function of this modified estimator may be written as 
\begin{align}
    & \int_{(0,t]} \phi_u -  \left\{ \mathbb{E}(\phi_u \,\mid\, L, A, C) - \mathbb{E} (\phi_u \,\mid\, L) \right\} \nonumber \\
    =\, & \int_{(0,t]}\phi_u
    - \left\{ \mathbb{E} (\phi_u \,\mid\, L, A) - \mathbb{E} (\phi_u \,\mid\, L) \right\} 
    -  \left\{ \mathbb{E} (\phi_u \,\mid\, L, A, C) - \mathbb{E} (\phi_u \,\mid\, L, A) \right\}. \label{eq:mod-su-if}
\end{align}
The terms in curly braces serve as augmentation terms; the first term
in curly braces in \cref{eq:mod-su-if} is
\begin{align*}
    \int_{(0,t]} \mathbb{E} (\phi_u \,\mid\, L, A,C) - \mathbb{E} (\phi_u \,\mid\, L)
    & = \int_{(0,t]} \left[ \frac{I(A=a)}{\pi_0(a; L)} \d \E \{ N^*_a (u) \mid A=a, L\} - \d \E \{ N^*_a (u) \mid A=a, L\} \right]\\
    & = \frac{I(A=a) - \pi_0(a; L)}{\pi_0(a; L)} \E \{N^*_a (t) \mid A=a, L\}.
\end{align*}

The second term in curly braces in \cref{eq:mod-su-if} is
\begin{align}
    \int_{(0,t]} \mathbb{E} (\phi_u \,\mid\, L, A, C) - \mathbb{E} (\phi_u \,\mid\, L, A) 
    & = \frac{I(A=a)}{\pi_0(a; L)} \int_{(0,t]} \left\{ \frac{I(C \geq u)}{K_0(u-; a, L)} - 1 \right\} \d \E\{ N^*_a (u) \mid A=a, L\} \nonumber \\
    & = \frac{I(A=a)}{\pi_0(a; L)} \left\{ \int_{(0,t]} \frac{I(C \geq u)}{K_0(u-; a, L)} \d \mu^*_a(u; a, L) - \mu^*_a(t; a, L) \right\}, \label{eq:mod-su-if2}
\end{align}
where $\mu^*_a(u; a, L) = \E\{ N^*_a (u) \mid A=a, L\}.$

Using integration by parts, the first summand in \cref{eq:mod-su-if2} may be written as
\begin{align*}
    \int_{(0,t]} \frac{I(C \geq u)}{K_0(u-; a, L)} \d \mu^*_a(u; L)
    & = \int_{(0,C \wedge t]} \frac{\d \mu^*_a(u; L)}{K_0(u-; a, L)} \\
    & = \frac{\mu^*_a(C \wedge t; L)}{K_0(C \wedge t; a, L)} 
    - \int_{(0,C \wedge t]} \mu^*_a(u; L) \d \left\{ \frac{1}{K_0(u; a, L)} \right\} \\
    & = \frac{\mu^*_a(C \wedge t; L)}{K_0(C \wedge t; a, L)} 
    - \int_{(0,t]} \mu^*_a(u; L) \frac{I(C \geq u)\d \Lambda_{C,0}(u; a, L)}{K_0(u; a, L)},
\end{align*}
where $\Lambda_{C,0}(u; a, L)$ is the cumulative hazard
corresponding to $K_0(u; a, L).$ Now, since
\begin{align*}
    \int_{(0,t]}  \mu_a^*(u; a, L) \frac{\d N_C(u)}{K_0(u; a, L)}
    = I(C \leq t) \frac{\mu^*_a(C; a, L)}{K_0(C; a, L)},
\end{align*}
we may add and subtract to rewrite \cref{eq:mod-su-if2} as
\begin{align*}
    & \frac{I(A=a)}{\pi_0(a; L)} \int_{(0,t]}
    \mu^*_a(u; L) \frac{\d N_C(u) - I(C \geq u)\d \Lambda_{C,0}(u; a, L)}{K_0(u; a, L)}  \\
    & \hspace{35mm} + \frac{I(A=a)}{\pi_0(a; L)} \left\{ \frac{\mu^*_a(C \wedge t; L)}{K_0(C \wedge t; a, L)} - I(C \leq t) \frac{\mu^*_a(C; a, L)}{K_0(C; a, L)} - \mu^*_a(t; L) \right\} \\
    =\, & \frac{I(A=a)}{\pi_0(a; L)} \int_{(0,t]}
    \mu^*_a(u; L)\frac{\d N_C(u) - I(C \geq u)\d \Lambda_{C,0}(u; a, L)}{K(u; a, L)}  \\
    & \hspace{35mm} + \frac{I(A=a)}{\pi(a; L)} \mu^*_a(t) \left\{ \frac{I(C>t)}{K_0(t; a, L)} - 1 \right\} \\
    =\, & \frac{I(A=a)}{\pi_0(a; L)} \int_{(0,t]} 
    \left\{ \mu^*_a(u; L) - \mu^*_a(t; L) \right\} \frac{\d N_C(u) - I(C \geq u)\d \Lambda_{C,0}(u; a, L)}{K_0(u; a, L)},
\end{align*}
where the last equality uses the (no failure) definition $N_C(u) = I(C \leq u)$ and applies the Duhamel equation \citep[Page 91]{andersen1993statistical}.
Using the fact that $\mu^*_a(u; a, L) = F_0(0,u;a,L)$ (i.e., calculated in the absence of failure, or assuming that $T^*_a = \infty$) the augmentation terms in curly braces in \cref{eq:mod-su-if} simplify to
\begin{align}
\label{eq:aug1}
-\frac{I(A=a) - \pi_0(a; L)}{\pi_0(a; L)} &
F_0(0,t;a,L) \\
\label{eq:aug2}
& 
- \frac{I(A=a)}{\pi_0(a; L)} \int_{(0,t]} 
    \left\{ F_0(0,u;a,L) - F_0(0,t;a,L) \right\} \frac{\d \tilde M_{C,0}(u;a,L) }{K_0(u; a, L)},
\end{align}
where
$
\d \tilde M_{C,0}(u;a,L) = \d N_C(u) - I(C \geq u) \d \Lambda_{C,0}(u; a, L).
$

Turning to the leading term in \cref{eq:mod-su-if},
we have
\begin{align}
\nonumber
\int_{(0,t]} \phi_u
& = \frac{I(A=a)}{\pi_0(a; L)} \int_{(0,t]} \frac{I(C \geq u)}{K_0(u-; a, L)} \d N (u) \\
\nonumber
& = \frac{I(A=a)}{\pi_0(a; L)} \int_{(0,t]} \frac{I(C \geq u)}{K_0(u-; a, L) \d N^*_a(u)} \\
\label{eq:lead}
& = \frac{I(A=a)}{\pi_0(a; L)} N^*_a(t)
+ \frac{I(A=a)}{\pi_0(a; L)} \int_{(0,t]}
\left\{ \frac{I(C \geq u)}{K_0(u-; a, L)} - 1 \right\} \d N^*_a(u).
\end{align}
Following arguments similar to those used
in proving \cref{lem:rr-lems}, it can be shown
that
\[
\frac{I(C \geq u)}{K_0(u-; a, L)} - 1 
= -\int_{(0,u)} \frac{\d \tilde M_{C,0}(u;a,L)}{K_0(u; a, L)};
\]
upon substitution for the term
in square brackets on the right-hand side 
of \cref{eq:lead}
and using integration by parts, we have
\begin{align*}
\frac{I(A=a)}{\pi_0(a; L)} & \int_{(0,t]} 
\left[ \frac{I(C \geq u)}{K_0(u-; a, L)} - 1 \right] \d N^*_a(u) \\
& = \frac{I(A=a)}{\pi_0(a; L)}
\left[ 
\int_{(0,t]} 
N^*_a(u) \frac{\d \tilde M_{C,0}(u;a,L)}{K_0(u; a, L)}
- N^*_a(t) \int_{(0,t]} 
\frac{\d \tilde M_{C,0}(u;a,L)}{K_0(u; a, L)}
\right].
\end{align*}
Combining this last expression with \eqref{eq:lead},
\begin{align*}
\int_{(0,t]} \phi_u
= N^*_a(t) & +
\frac{I(A=a)-\pi_0(a; L)}{\pi_0(a; L)} N^*_a(t) \\
& + \frac{I(A=a)}{\pi_0(a; L)} 
\left[ 
\int_{(0,t]} 
N^*_a(u) \frac{\d \tilde M_{C,0}(u;a,L)}{K_0(u; a, L)}
- N^*_a(t) \int_{(0,t]} \frac{\d \tilde M_{C,0}(u;a,L)}{K_0(u; a, L)}
\right].
\end{align*}

Substituting this expression, \eqref{eq:aug1} and \eqref{eq:aug2} into \eqref{eq:mod-su-if} and centering by $\psi_0$ then gives the influence function
\begin{align*}
\eqref{eq:mod-su-if} = N^*_a(t) - \psi_0 + &
\frac{I(A=a)-\pi_0(a; L)}{\pi_0(a; L)} 
\left( N^*_a(t) - F_0(0,t;a,L) \right) \\
& + 
\frac{I(A=a)}{\pi_0(a; L)} 
\Biggl[
\int_{(0,t]} \left\{ 
N^*_a(u) - F_0(0,u;a,L) \right\} 
\frac{\d \tilde M_{C,0}(u;a,L)}{K_0(u; a, L)} \\
& 
- \left\{ N^*_a(t) - F_0(0,t;a,L) \right\} 
\int_{(0,t]} \frac{\d \tilde M_{C,0}(u;a,L)}{K_0(u; a, L)} \Biggr].
\end{align*}
The resulting influence function \eqref{eq:mod-su-if} agrees with that stated in \cref{prop:mod-su-IF0}.

Respectively replacing $\pi_0(a,L),$ $F_0(u,t;a,L)$ and $K_0(u;a,L),$with models $\pi(a,L),$ $F(u,t;a,L)$ and $K(u;a,L),$ 
the indicated double robustness property follows from adding and subtracting 1 from the factor $\frac{I(A=a)}{\pi_0(a; L)}$ that
multiplies the last term in square brackets and then using the fact that the uncentered influence function is the relevant  augmented estimator of $\psi_0$. 
\end{proof}

\begin{proof}[Proof of \cref{prop:mod-su-IF1}]
Allowing for the possibility of failure, we 
first consider the general form of the influence function 
in \cref{prop:class}:
\begin{align}
\label{eq:su-dr-1}
 \frac{I(A=a)}{\pi_0(a; L)} 
      \frac{\Delta}{K_0(X-;a,L)} & N(t)
      - \psi_0 
       - \Bigl\{ I(A=a) - \pi_0(a; L) \Bigr\} h_1(L) \\
       &
\label{eq:su-dr-2}       
       + \int_{(0,\infty)} h_2 \{ u; \bar{N}(u), A, L \} \,
      \d M_{C,0}(u; A, L),
    \end{align}
where $
\d M_{C,0}(u;a,L) = \d N_C(u) - Y^\dag(u) \d \Lambda_{C,0}(u; a, L)$ and $N_C(u) = I\{ X \geq u, \Delta = 0\}.$

Making the same choices as in \cref{prop:mod-su-IF0}, which at the moment allows for the possibility of failure, we have
$h_1(L)  = F_0(0, t; a, L) / \pi_0(a; L)$ and
\[
        h_2(u; A, L, \bar{N}(u))  = \frac{I(A=a)}{\pi_0(a; L)} \frac{N(u\wedge t) - I(u \leq t) \left\{ F_0(0, u; a, L) - F_0(0, t; a, L) \right\}}{K_0(u; a, L)}.
\]
Equation \eqref{eq:su-dr-1} can then be written
\begin{align*}
N^*_a(t) & -  \psi_0 +  \frac{I(A=a)-\pi_0(a; L)}{\pi_0(a; L)} 
\left( N^*_a(t) - F_0(0, t; a, L) \right) 
+ 
\frac{I(A=a)}{\pi_0(a; L)}
\left( \frac{\Delta}{K_0(X-;a,L)} - 1 \right) N(t);
\end{align*}
the last term in this expression can be equivalently expressed as
\begin{align}
\label{eq:su-dr-1-last}
-\frac{I(A=a)}{\pi_0(a; L)}
N(t) \int_{(0,\infty)} & \frac{\d M_{C,0}(u; a, L)}{K_0(u;a,L)}.
\end{align}
Equation \eqref{eq:su-dr-2} can be written
\begin{align*}
\frac{I(A=a)}{\pi_0(a; L)} &
\int_{(0,\infty)} 
N(u\wedge t) \frac{\d M_{C,0}(u; a, L)}{K_0(u;a,L)}
  -
\frac{I(A=a)}{\pi_0(a; L)}
\int_{(0,t]} \hspace{-5mm}
\left\{ F_0(0, u; a, L) - F_0(0, t; a, L) \right\}
\frac{\d M_{C,0}(u; a, L)}{K_0(u;a,L)} \\
& = \frac{I(A=a)}{\pi_0(a; L)}
\biggl(
\int_{(0,t]} 
N(u) \frac{\d M_{C,0}(u; a, L)}{K_0(u;a,L)}
+
N(t) 
\int_{(t,\infty)} \frac{\d M_{C,0}(u; a, L)}{K_0(u;a,L)} \\
& 
\hspace{5mm} -
\int_{(0,t]} 
F_0(0, u; a, L) 
\frac{\d M_{C,0}(u; a, L)}{K_0(u;a,L)}
+ F_0(0, t; a, L)
\int_{(0,t]} 
\frac{\d M_{C,0}(u; a, L)}{K_0(u;a,L)}
\biggr).
\end{align*}
Combining \eqref{eq:su-dr-1-last} and \eqref{eq:su-dr-2} and simplifying, we obtain
\[
\frac{I(A=a)}{\pi_0(a; L)}
\left\{
\int_{(0,t]} 
\Bigl( N(u) - F_0(0, u; a, L) \Bigr) \frac{\d M_{C,0}(u; a, L)}{K_0(u;a,L)}
- 
\Bigl( N(t) - F_0(0, t; a, L) \Bigr) \int_{(0,t]} 
 \frac{\d M_{C,0}(u; a, L)}{K_0(u;a,L)}
\right\}.
\]
Putting things together, the sum of \eqref{eq:su-dr-1} and \eqref{eq:su-dr-2} reduces to
\begin{align*}
N^*_a(t) & -  \psi_0 +  \frac{I(A=a)-\pi_0(a; L)}{\pi_0(a; L)} \left( N^*_a(t) - F_0(0, t; a, L) \right) + \frac{I(A=a)}{\pi_0(a; L)} \times \\
& 
\left\{
\int_{(0,t]} 
\Bigl( N(u) - F_0(0, u; a, L) \Bigr) \frac{\d M_{C,0}(u; a, L)}{K_0(u;a,L)}
- 
\Bigl( N(t) - F_0(0, t; a, L) \Bigr) \int_{(0,t]} 
 \frac{\d M_{C,0}(u; a, L)}{K_0(u;a,L)}
\right\}.
\end{align*}
The augmentation term in the second line of this last expression nominally depends on the observed
recurrent event process. Using the definition of  $N^*(\cdot),$ the fact that $Y(u) N(u) = Y(u) N^*(u),$ and in view of
the fact that $Y(u) Y^\dag(u) = Y^\dag(u),$ it can be shown for an arbitrary function $h(N(u),a,L)$ that
\[
I ( A = a ) \int_0^t h(N(u),a,L)  \d M_{C,0}(u; a, L) =
I (A = a ) \int_0^t h(N^*_a(u),a,L) 
\d M_{C,0}(u; a, L).
\]

Consequently, the influence function defined by \eqref{eq:su-dr-1} and \eqref{eq:su-dr-2}
with the indicated choices of $h_1$ and $h_2$ can be rewritten
\begin{align*}
N^*_a(t) & -  \psi_0 +  \frac{I(A=a)-\pi_0(a; L)}{\pi_0(a; L)} \left( N^*_a(t) - F_0(0, t; a, L) \right) + \frac{I(A=a)}{\pi_0(a; L)} \times \\
& 
\left\{
\int_{(0,t]} 
\Bigl( N^*_a(u) - F_0(0, u; a, L) \Bigr) \frac{\d M_{C,0}(u; a, L)}{K(u;a,L)}
- 
\Bigl( N^*_a(t) - F_0(0, t; a, L) \Bigr) \int_{(0,t]} 
 \frac{\d M_{C,0}(u; a, L)}{K_0(u;a,L)}
\right\}.
\end{align*}
In the event that failure is impossible, the term $\d M_{C,0}(u; A, L)$  reduces to $\d \tilde M_{C,0}(u; A, L),$ 
and the above expression continues to hold with suitable interpretations of $N^*_a(t)$ and $F_0.$ The resulting
final expression is exactly that given in \eqref{eq:mod-su-if}.
\end{proof}

\section{Additional Details and Results of Numerical Studies} \label{sup:simulation}

In all scenarios, death times $T$ were generated with a proportional hazards model
\[\lambda_T(t; A, L) = \rho_{0,T}(t)\exp\{\beta_T^{(1)} + \beta_T^{(2)}A + \beta_T^{(3)}L_1 + \beta_T^{(4)}L_2 + \beta_T^{(5)}L_3 + \beta_T^{(6)}AL_3 + \beta_T^{(7)}L_1L_2\},\] with $\beta_T = (-2, -0.5, 0.1, 0.1, -0.5, -0.3, 0.1) ^{\top}$. The baseline hazard $\rho_{0,T}(t)$ is the hazard of a Weibull distribution with scale parameter $\lambda_T = 1$ and shape parameter $k_T = 1.1$. The death process for $T$ was truncated at $\tau = 12$, yielding $T' = T \wedge \tau$. 

The recurrent events were generated with the intensity model 
\[\lambda_E(t; A, L) = \rho_{0,E}(t)\exp\{\beta_E^{(1)} + \beta_E^{(2)}A + \beta_E^{(3)}L_1 + \beta_E^{(4)}L_2 + \beta_E^{(5)}L_3 + \beta_E^{(6)}AL_2 + \beta_E^{(7)}L_1L_3\},\]
with $\beta_E = (1, -0.5, 0.1, 0.1, -0.5, -0.1, -0.5)^{\top}$. The baseline hazard $\rho_{0,E}(t)$ is the hazard of a Weibull distribution with scale parameter $\lambda_E = 1$ and shape parameter $k_E = 1.1$. The recurrent events were truncated at $X = \min\{T', C\}$, yielding $N(t) = N^{**}\{\min\{t, T', C\}\}$. We can see that $L_3$ is an important covariate which reduces the number of events and increases survival probability. 

The censoring times $C$ were generated using the proportional hazards model
\[\lambda_C(t;A,L) = \rho_{0,C}(t)\exp\{\beta_C^{(1)} + \beta_C^{(2)}A + \beta_C^{(3)}L_1 + \beta_C^{(4)}L_2 + \beta_C^{(5)}L_3\},\]
with baseline hazard $\rho_{0,C}(t)$ being the hazard of a Weibull distribution with scale parameter $\lambda_C = 1$ and shape parameter $k_C = 1$. The censoring processes were set to terminate at $10^6$, which means that there is effectively no right censoring for $C$.  

The recurrent events, death time $T$ and censoring $C$ were all generated using the function \texttt{simEventData()} from the \texttt{reda} package in \texttt{R}. The coefficients for the propensity score models and censoring models in each simulation scenario are provided in Table~\ref{tab:sim-scenario-param}.

\begin{table}[h!]
    \centering
    \begin{tabular}{c|cc}
    \hline
    Scenario & Propensity score model $\pi(\cdot)$ & Censoring model $K(\cdot)$ \\
    \hline
    1 & $\beta_A = (-0.5, 0.1, 0.1, 0.5)$ & $\beta_C = (-2, 0.5, 0, 0, 0)$ \\
    2 & $\beta_A = (-0.5, 0.1, 0.1, 0.5)$ & $\beta_C = (-3, 0, 0.1, 0.1, 0.5)$ \\
    3 & $\beta_A = (-0.5, 0, 0, 0.5)$ & $\beta_C = (-3, 0, 0.1, 0.1, 0.5)$ \\
    \hline \\
    \end{tabular}
    \caption{Parameters for the propensity score and censoring models in the three simulation scenarios.}
    \label{tab:sim-scenario-param}
\end{table}

We calculated our one-step AIPW estimators as follows. For the propensity score model $\pi(\cdot)$, {we used generalized linear model (GLM) from \texttt{base R}, lightGBM \citep{ke2017lightgbm} from the \texttt{lightgbm} package, and random forest from the \texttt{randomForest} package in a SuperLearner model trained using 10-fold cross validation using the \texttt{SuperLearner} package. The survival SuperLearner, proposed by \cite{westling2021inference}, was used to estimate the conditional censoring and event models with libraries containing survival random forest \citep{ishwaran2008random} from the \texttt{randomForestSRC} package, Cox proportional hazards model from the \texttt{survival} package, and a Cox proportional hazards generalized additive model from the \texttt{mgcv} package.} The survival SuperLearner (from the \texttt{survSuperLearner} package) were trained on 10-fold cross validation, with survival random forest as the initial estimator, and limiting the recursive SuperLearner coefficient fitting to 20 iterations. 

Note that we also have to estimate the nuisance parameters $F(u,t;a,l)$ for each landmark times $t \in \{1, 2, 3, 4, 5, 6\}$ on a grid of times $u$. We considered such a grid which ranges from 0 to 12 with 0.01 spacing. Note that
\begin{align*}
    F(u,t;a,l) & = \mathbb{E}\left\{\frac{\Delta}{K(X-;A,L)}I(X>u)N(t) \mid A = a, L = l\right\} \\
    & = b(u,t;a,l)\Pr(X > u \mid A=a, L=l) \\
    & = b(u,t;a,l)K(u;a,l)H(u;a,l),
\end{align*}
where
\[b(u,t;a,l) = \mathbb{E}\left\{\frac{\Delta}{K(X-;A,L)}N(t) \mid X > u, A = a, L = l\right\}\]
Furthermore, we also have 
$b(u,t;a,l) = c(u,t;a,l) \times d(u;a,l)$ where
\begin{align*}
    c(u,t;a,l) & = \mathbb{E}\left\{\frac{N(t)}{K(X-;A,L)} \mid \Delta = 1, X > u, A = a, L = l\right\} \\
    d(u;a,l) & = \mathbb{E}\left\{\Delta \mid X > u, A = a, L = l\right\}
\end{align*}
Therefore, to estimate $F(u,t;a,l)$ on the $u$-grid, we effectively need to estimate $c(u,t;a,l)$ and $d(u;a,l)$ for every time $u$ on the grid. Since the grid $u$ is dense, separate estimation of $c(u,t;a,l)$ and $d(u;a,l)$ for each $u$ is computationally burdensome, we considered an approximate strategy with another grid of times $s$ being the $\{0, 0.05, 0.1, 0.15, ..., 0.95\}$-quantile of the $u$ grid. Let $s_1, ..., s_{20}$ denote the $s$-grid points in increasing order. At each grid point $s_j$, $j = 1, ..., 20$, we fitted $c(s_j, t; a, l)$ and $d(s_j; a, l)$ models on a subset of the data that satisfy $\{\Delta = 1, X > s_j\}$ and $\{X > s_j\}$, respectively. For grid points $s_j$ where the subset of data has less than 10 observations, we use the models for $s_{j-1}$ as estimates for models at $s_j$. The $c(\cdot)$ and $d(\cdot)$ models were trained using 10-fold cross-validated SuperLearner (from the \texttt{SuperLearner} package) with libraries including the {generalized linear model (GLM) from \texttt{base R}, lightGBM \citep{ke2017lightgbm} from the \texttt{lightgbm} package, and random forest from the \texttt{randomForest} package.} The Gaussian family was used for $c(\cdot)$ while the binomial family was used for $d(\cdot)$. Finally, to make estimations of $F(\cdot)$ on the $u$-grid, we use linear interpolation of the estimates on the $s$-grid.

For the IPW estimators, we used a logistic regression (from \texttt{base R}) and a CoxPH model (from the \texttt{survival} package) containing main effects of the covariates (and the treatment assignment for the CoxPH model) to estimate the propensity score and censoring probabilities. The propensity score in the double inverse weight estimate $\hat{\mu}^{\text{SZ}}_a(t)$ was also estimated using a logistic regression with main effects of the covariates. Finally, we used the default settings of the \texttt{CFsurvival} package by \cite{westling2021inference} to obtain their implemented estimates for $\eta_a^*(t)$.

We include the simulation results for each landmark time $t \in \{1, 2, 3, 4, 5, 6\}$, which show similar trends to the average results in the main text.

\begin{figure}[htbp!]
    \centering
    \includegraphics[width=\linewidth]{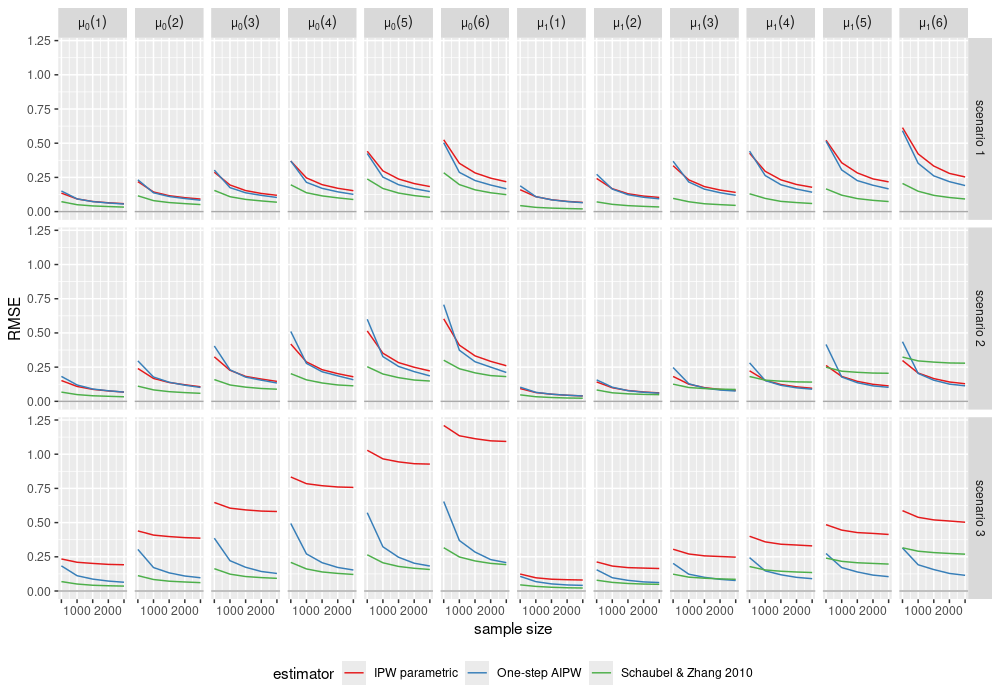}
    \caption{Point-wise root mean squared error (RMSE) of the estimators for $\mu_1^*(t)$, $\mu_0^*(t)$ in three simulation scenarios as a function of the sample size $n$.}
    \label{fig:rmse-mu}
\end{figure}

\begin{figure}[htbp!]
    \centering
    \includegraphics[width=\linewidth]{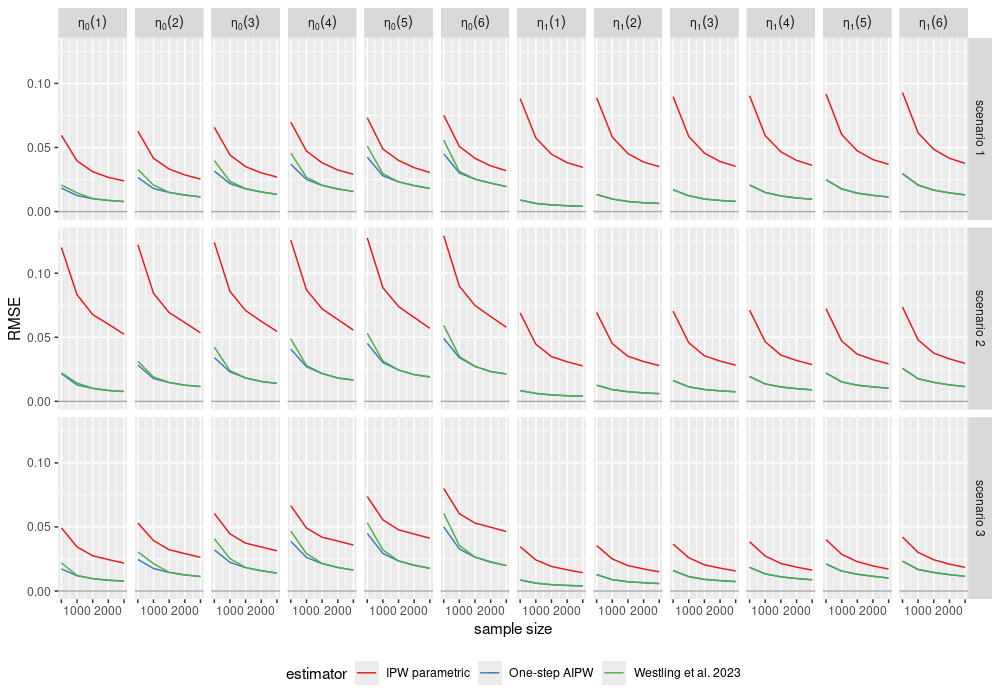}
    \caption{Point-wise root mean squared error (RMSE) of the estimators for $\eta_1^*(t)$, $\eta_0^*(t)$ in three simulation scenarios as a function of the sample size $n$. .}
    \label{fig:rmse-eta}
\end{figure}

\begin{figure}[htbp!]
    \centering
    \includegraphics[width=\linewidth]{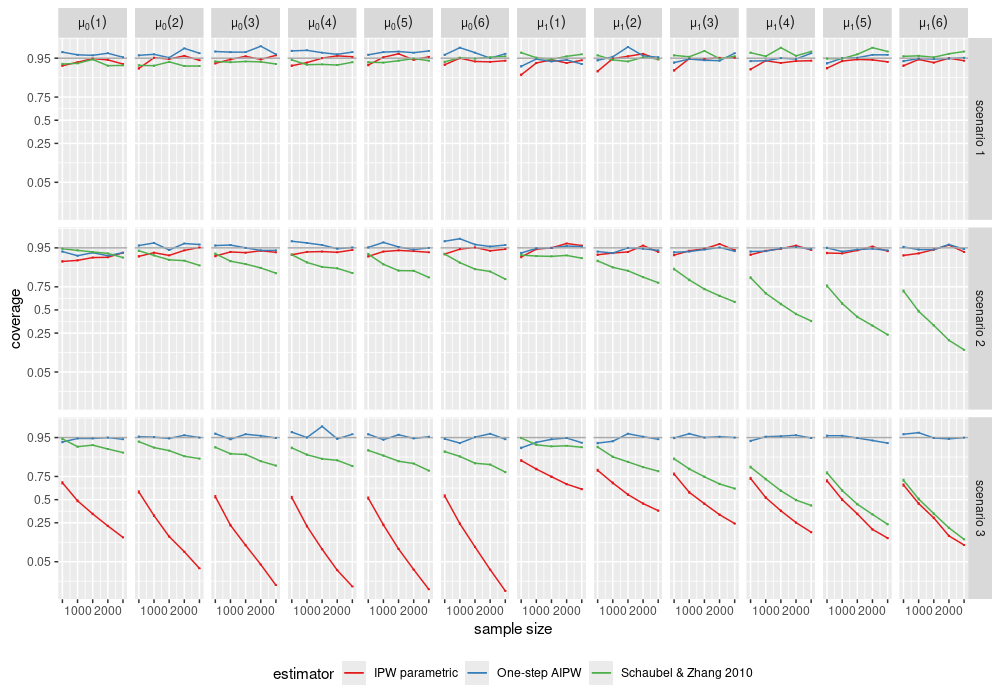}
    \caption{Coverage of the estimators for $\mu_1^*(t)$, $\mu_0^*(t)$ in three simulation scenarios as a function of the sample size $n$. The vertical axis is plotted in the logistic scale, and the tick labels indicate values in the original scale. The error bars indicate 95\% confidence intervals considering uncertainty due to the finite number of simulation replications.}
    \label{fig:coverage-mu}
\end{figure}

\begin{figure}[htbp!]
    \centering
    \includegraphics[width=\linewidth]{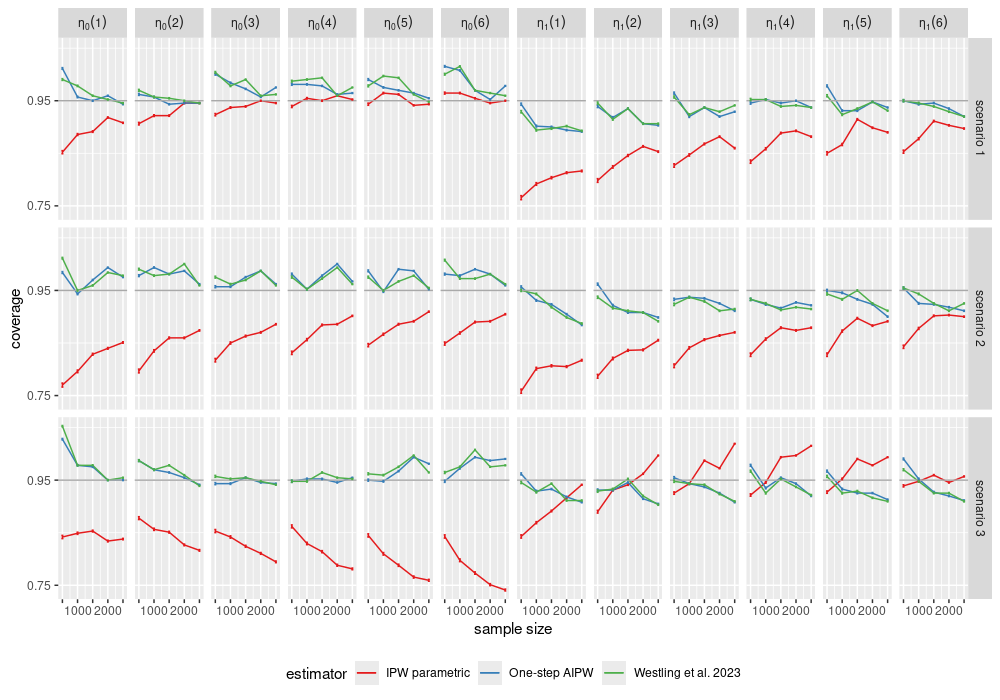}
    \caption{Coverage of the estimators for $\eta_1^*(t)$, $\eta_0^*(t)$ in three simulation scenarios as a function of the sample size $n$. The vertical axis is plotted in the logistic scale, and the tick labels indicate values in the original scale. The error bars indicate 95\% confidence intervals considering uncertainty due to the finite number of simulation replications.}
    \label{fig:coverage-eta}
\end{figure}

\newpage
\section{Additional Details and Results for Data Application} \label{sub-application}

We considered a dynamic cohort of Medicare Fee for Service (FFS) beneficiaries who turned 65 and resided in Arizona, United States, during the period of 2000-2012. We followed the cohort for a maximum of four years after they turned 65, in which the first two years were treated as the baseline period, where baseline covariates were measured. The cohort's health information, including hospitalizations and possibly death, was observed during the last two years of follow-up (67-69 years old). Therefore, the administrative censoring in our study is 2 years or 24 months. We excluded individuals who lived in multiple ZIP codes during the four years of possible follow-up because the month and year of the move were unknown. Individuals with missing covariates (which we will discuss next) or who passed away during the baseline period were also excluded from the study. This left us with a cohort of 272,226 individuals, whose demographics are summarized in Table~\ref{tab:cohort-demo}.

\begin{table}[!ht]
\centering
\tiny
\begin{tabular}{rrrrrrrrrrrrrr}
  \toprule[1pt]
Turned 65 in & 2000 & 2001 & 2002 & 2003 & 2004 & 2005 & 2006 & 2007 & 2008 & 2009 & 2010 & 2011 & 2012 \\ 
  \midrule
N & 14528 & 15699 & 17582 & 18833 & 18624 & 18675 & 19244 & 20004 & 22687 & 22724 & 25780 & 26124 & 31722 \\ 
\multicolumn{1}{l}{\textbf{Sex}}\\
Male & 7173 & 7851 & 8787 & 9378 & 9188 & 9228 & 9466 & 9782 & 11155 & 11179 & 12589 & 13021 & 15901 \\ 
Female & 7355 & 7848 & 8795 & 9455 & 9436  & 9447 & 9778 & 10222 & 11532 & 11545 & 13191 & 13103 & 15821 \\
\addlinespace
\multicolumn{1}{l}{\textbf{Race}}\\
White & 13203 & 14294 & 15908 & 17091 & 16924 & 16893 & 17264 & 17958 & 20385 & 20432 & 23072 & 23015 & 27788 \\ 
Black & 262 & 251 & 362 & 330 & 352 & 346 & 385 & 380 & 419 & 421 & 485 & 502 & 607 \\ 
Asian & 78 & 83 & 97 & 98 & 107 & 123 & 146 & 144 & 154 & 188 & 217 & 243 & 294 \\ 
Hispanic & 223 & 284 & 311 & 328 & 285 & 314 & 359 & 407 & 446 & 437 & 543 & 603 & 642 \\ 
North American Native & 387 & 380 & 642 & 699 & 680 & 721 & 764 & 772 & 855 & 835 & 911 & 976 & 1161 \\
Other/Unknown & 375 & 407 & 262 & 287 & 276 & 278 & 326 & 343 & 428 & 411 & 552 & 785 & 1230 \\ 
\addlinespace
\multicolumn{6}{l}{\textbf{Terminal events or censoring during study period}}\\
Died age $<69$ & 550 & 619 & 701 & 622 & 613 & 594 & 611 & 620 & 705 & 689 & 798 & 843 & 1035 \\
Censored age $<69$ & 797 & 1128 & 2635 & 3334 & 2250 & 2100 & 1949 & 1757 & 2070 & 2079 & 2486 & 2770 & 3015 \\
   \bottomrule[1pt]\\[5ex]
\end{tabular}
\caption{Demographic summaries of the study cohort.}
\label{tab:cohort-demo}
\end{table}

We collected individual-level data on sex and race (categorized as White, Black, Asian/Pacific Islander, Hispanic, North American Native, or Other) from the Medicare Master Beneficiary Summary File (MBSF). Annual demographic and socio-economic characteristics at the ZIP-code level were obtained from the 2000 U.S. Census and the American Community Survey (ACS) for the years 2011-2016. These variables included population density; the percentages of residents identifying as Black; the percentages of residents identifying as Hispanic; the percentage of high school graduates; median home value; median household income; the proportion of owner-occupied housing; and the percentage of people living below the federal poverty line. We also obtained annual ZIP-code-level of average body mass index and the percentage of smoking population from the Center for Disease Control’s Behavior Risk Factor Surveillance System surveys from 2000-2012. When data were missing for a specific year and ZIP code, we used linear interpolation for imputation. Finally, the baseline ZIP-code-level covariates for each individual were calculated by averaging the values over their respective baseline period. Because our cohort was dynamic, with individuals entering as they turned 65, we also included the year each person turned 65 as a baseline covariate to account for temporal effects.

To estimate PM2.5 exposure, we used daily predictions of PM2.5 concentrations at a 1 $km^2$ spatial resolution, generated by an ensemble learning model with a cross-validated $R^2$ of 0.89 \citep{Di2019AnResolution}.  Monthly exposure levels were then determined by averaging the predicted values from all grid centroids located within each ZIP code boundary. For each individual, baseline exposure was defined as the average of these monthly PM2.5 concentrations at their residential ZIP code during their baseline period. Figure~\ref{fig:AZ-pm25} shows the distribution of the baseline PM2.5 exposure among our study cohort in Arizona, We can see that the distribution is bimodal with high-exposure distributed around Phoenix, a metropolis and Yuma, the home of the Yuma Proving Ground, a series of environmentally specific test centers for the U.S. army.

\begin{figure}[htbp!]
    \centering
    \includegraphics[width=0.45\linewidth]{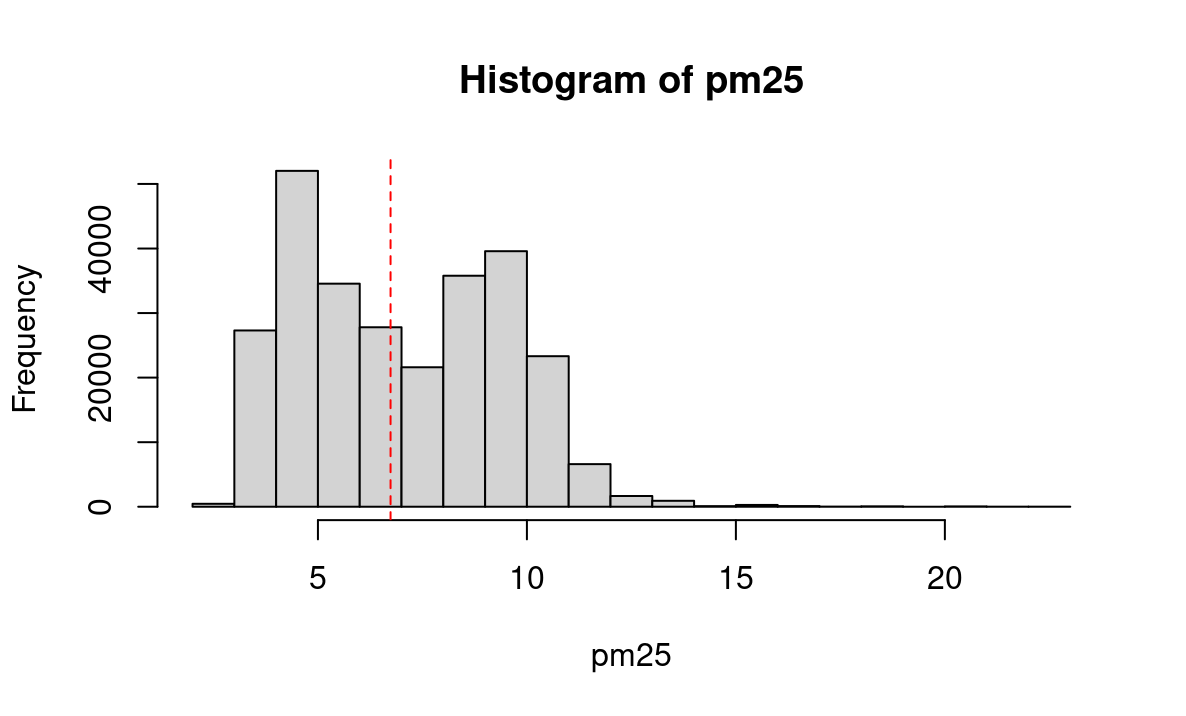} \includegraphics[width=0.45\linewidth]{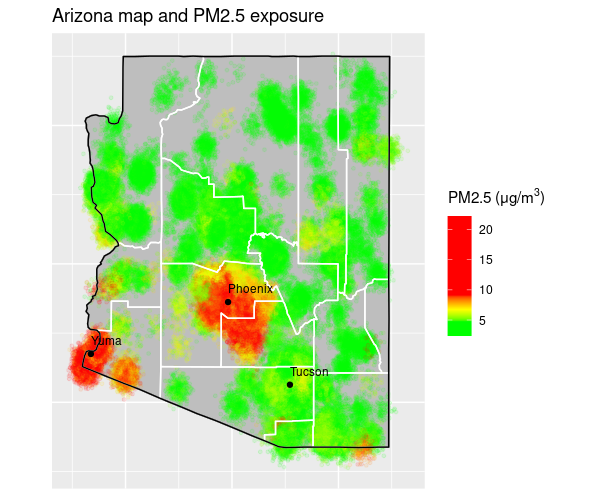}
    \caption{Distribution of PM2.5 exposure during the baseline period among the study cohort in Arizona.}
    \label{fig:AZ-pm25}
\end{figure}

For health outcomes, we retrieved the dates of death from the MBSF and hospital admission details from the Medicare Part A MedPAR dataset. Hospital admissions were classified as cardiovascular disease (CVD)-related if the billing diagnosis during the hospitalization included any of the following: atrial fibrillation (ICD-9: 427.3; ICD-10: I48), cardiac arrest (ICD-9: 427.5; ICD-10: I46), or acute myocardial infarction (ICD-9: 410; ICD-10: I21). These diagnoses were based on the International Classification of Diseases (ICD) codes, with revisions ICD-9 used for the years 2000-2015. When multiple hospitalizations occurred within a two-day window (such as transfers between hospitals), only the first admission was counted as a hospitalization event. Individuals were considered censored if any covariate data were missing after the baseline period.

After the data is obtained, we applied our method to study the impact of a 2-year exposure to a higher level of PM2.5 compared to a lower level. In the main text, we defined the high exposure group as individuals who are exposed to the top 25\% of baseline PM2.5 levels. Similarly, we defined the low exposure group as individuals who are exposed to the bottom 25\%. The results for other definitions of high vs. low exposure groups are presented below in Figures~\ref{fig:AZ-35} and \ref{fig:AZ-50} for sensitivity assessment. 
 
\begin{figure}[h!]
    \centering
    \includegraphics[width=\linewidth]{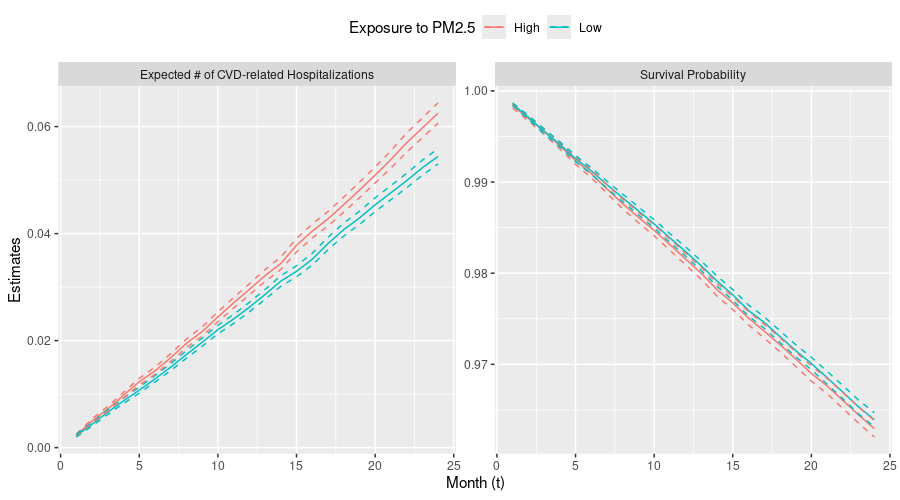}
    \includegraphics[width=\linewidth]{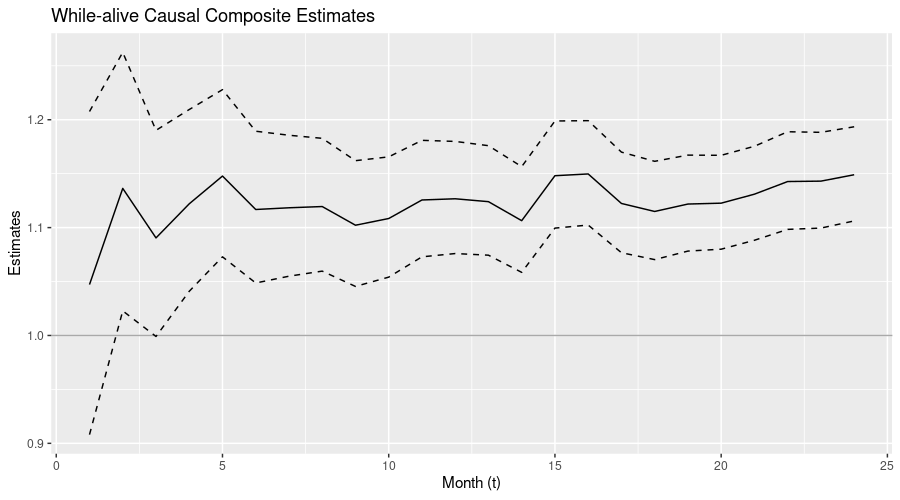}
    \caption{The data analysis results where the high exposure group is defined as being exposed to the top 35\% of the PM2.5 level (PM2.5 $\ge 8.42$ $\mu g/m^3$) and the low exposure group is defined as being exposed to the bottom 35\% of the PM2.5 level (PM2.5 $\le 5.42$ $\mu g/m^3$).}
    \label{fig:AZ-35}
\end{figure}

\begin{figure}[h!]
    \centering
    \includegraphics[width=\linewidth]{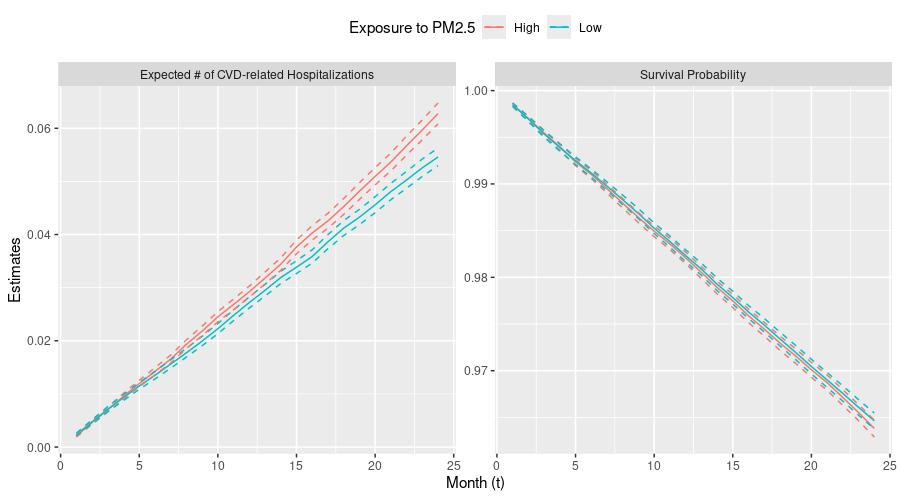}
    \includegraphics[width=\linewidth]{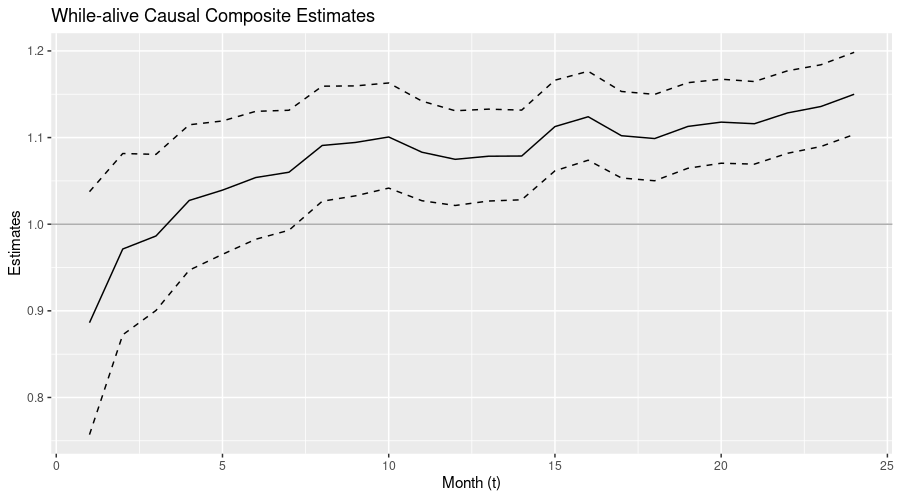}
    \caption{The data analysis results where the high exposure group is defined as being exposed to the top 50\% of the PM2.5 level (PM2.5 $\ge 6.74$ $\mu g/m^3$) and the low exposure group is defined as being exposed to the bottom 50\% of the PM2.5 level (PM2.5 $< 6.74$ $\mu g/m^3$).}
    \label{fig:AZ-50}
\end{figure}

\newpage
\clearpage
\putbib[references]
\end{bibunit}